\documentclass[aps,pra,twocolumn,showpacs,amsmath,amssymb,floatfix,superscriptaddress,notitlepage]{revtex4-1}

\usepackage{color}
\usepackage{bbm}
\usepackage{graphicx}
\usepackage{dcolumn}
\usepackage[utf8]{inputenc}

\usepackage{graphicx}
\usepackage{epstopdf}
\usepackage{amsthm}
\usepackage{amsmath}
\usepackage{empheq}
\usepackage{bbm}
\usepackage{braket}
\usepackage{amssymb}
\usepackage[usenames,dvipsnames]{xcolor}
\usepackage{cases}
\usepackage{latexsym}
\usepackage[colorlinks=true,citecolor=Cerulean,linkcolor=RubineRed,urlcolor=Cerulean]{hyperref}

\graphicspath{ {images/}{images/figure1/}{images/figure2/}{images/figure3/}{images/figure4/}{images/figure5/} }

\newtheorem{theorem}{Theorem}

\newtheorem{lemma}{Lemma}

\newtheorem{definition}{Definition}

\newtheorem{innercustomgeneric}{\customgenericname}
\providecommand{\customgenericname}{}
\newcommand{\newcustomtheorem}[2]{%
	\newenvironment{#1}[1]
	{%
		\renewcommand\customgenericname{#2}%
		\renewcommand\theinnercustomgeneric{##1}%
		\innercustomgeneric
	}
	{\endinnercustomgeneric}
}

\newcustomtheorem{customdef}{Definition}

\newtheorem{innercustomthm}{Theorem}
\newenvironment{customthm}[1]
  {\renewcommand\theinnercustomthm{#1}\innercustomthm}
  {\endinnercustomthm}

\newcommand{\alv}[1]{{\color{black} #1}}

\newcommand{\id}{\mathbbm{1}} 
\newcommand{\tr}[1]{\operatorname{\textnormal{Tr}}\left( {#1} \right)}  
\newcommand{\Ckubo}{\operatorname{ C_{\textnormal{Kubo}} } } 
\usepackage{lipsum}
\usepackage{graphicx}

\begin{document}

\title{Time evolution of correlation functions in quantum many-body systems}

\date{\today}

\author{\'Alvaro M. Alhambra}
\email{aalhambra@perimeterinstitute.ca }
\affiliation{Perimeter Institute for Theoretical Physics, Waterloo, ON N2L 2Y5, Canada}

\author{Jonathon Riddell}
\email{riddeljp@mcmaster.ca}
\affiliation{Department of Physics \& Astronomy, McMaster University
	1280 Main St.  W., Hamilton ON L8S 4M1, Canada.}

\author{Luis Pedro Garc\'ia-Pintos}
\thanks{Corresponding author}
\email{lpgp@umd.edu}
\affiliation{Department of Physics, University of Massachusetts, Boston, MA 02125, USA}
\affiliation{Joint Center for Quantum Information and Computer Science, NIST/University of Maryland, College Park, Maryland 20742, USA}
\affiliation{Joint Quantum Institute, NIST/University of Maryland, College Park, Maryland 20742, USA}

\begin{abstract}
We give rigorous analytical results on the temporal behavior of two-point correlation functions --also known as dynamical response functions or Green’s functions-- in closed many-body quantum systems. We show that in a large class of  translation-invariant  models the correlation functions factorize at late times $\langle A(t) B\rangle_\beta \rightarrow \langle A \rangle_\beta \langle B \rangle_\beta$, thus proving that dissipation emerges out of the unitary dynamics of the system. We also show that  for systems with a generic spectrum  the fluctuations around this late-time value are bounded by the purity of the thermal ensemble, which generally decays exponentially with system size.
For auto-correlation functions we provide an upper bound on the timescale at which they reach the factorized late time value. Remarkably, this bound is only a function of local expectation values, and does not increase with system size. We give numerical examples that show that this bound is a good estimate in non-integrable models, and argue that the timescale that appears can be understood in terms of an emergent fluctuation-dissipation theorem. 
Our study extends to further classes of two point functions such as the symmetrized ones and the Kubo function that appears in linear response theory, for which we give analogous results. 

\end{abstract}

\maketitle

Two-point correlation functions --or dynamical response/Green's functions-- are the central object of the theory of linear response \cite{kubo1957statistical}, and appear in the characterization of a wide range of non-equilibrium and statistical phenomena in the study of quantum many-body systems and condensed matter physics~\cite{rickayzen2013greenBook}.
This includes different types of scattering and spectroscopy experiments \cite{zhu2005modern}, quantum transport \cite{zwanzig1965time,zotos1997transport}, and fluctuation-dissipation relations \cite{Srednicki99,Khatami13,d2016quantum}.
They have also appeared in the characterization of topological \cite{nussinov2008autocorrelations} and crystalline ordering \cite{watanabe2015absence}, of quantum-chaotic systems~\cite{luitz2016anomalous,gharibyan2019characterization} and of different notions of ergodicity in quantum and classical systems \cite{cornfeld2012ergodic,venuti2019ergodicity}.

Here we study the time evolution of such correlation functions in isolated systems evolving under unitary dynamics. More precisely, we focus on functions of the form 
\begin{align}
C^{AB}(t)&\equiv \langle A(t) B \rangle_\beta =\tr{\rho  A(t) B},
\end{align}
where the evolution is generated by a time-independent Hamiltonian $H$, $\rho \equiv  e^{-\beta H} / Z_\beta$ is a thermal state at inverse temperature $\beta$ with partition function $Z_\beta$, and $A(t) = e^{i Ht}Ae^{-i Ht} $ is the evolved observable in the Heisenberg picture. Both $A$ and $B$ are usually taken to be either local (such as a single-site spin) or extensive operators (such as a global current or magnetization).

Two-point correlation functions have been widely studied before, mostly through numerical methods such as exact diagonalization \cite{steinigeweg2009density}, QMC~\cite{starykh1997dynamics} and tensor networks \cite{sirker2005real,barthel2009spectral,
feiguin2010spectral,karrasch2012finite,
dargel2012lanczos,barthel2013precise,
bohrdt2017scrambling}, and analytically 
 for specific models, e.g.  \cite{korepin1997quantum,
its1993temperature,
stolze1995gaussian,
reyes2006finite,perk2009new,
bertini2019exact}. Also, a number of experimental schemes to measure it directly have been proposed~\cite{romero2012quantum,
knap2013probing,
buscemi2013direct,gessner2014nonlinear,uhrich2017noninvasive}, which manage to circumvent the obstacle of having to measure two non-commuting observables on a single system.
Here, we give rigorous analytical results on their dynamical behavior with as few assumptions on the Hamiltonian as possible. 
Our results apply to most  translation-invariant  non-integrable Hamiltonians, in which the degeneracy of the energy spectrum is small.

First, for arbitrary local observables $A$ and $B$ we 
prove that, for late times, the following signature of dissipation occurs in a large class of translation-invariant  models 
\begin{equation} \label{eq:factor}
\langle A(t) B \rangle_\beta \stackrel{t\rightarrow \infty}{\longrightarrow} \langle A \rangle_\beta \langle B \rangle_\beta.
\end{equation}
Moreover, we show that the fluctuations around the late-time value are in fact
 bounded by the effective dimension $d^{-1}_{\text{eff}} \equiv\tr{\rho^2}$ of the ensemble, which decays quickly with system size.
 
For the 
 case of auto-correlation functions, when $A=B$, we also derive an upper bound on the timescale at which the factorization of Eq. \eqref{eq:factor} happens, which, remarkably, is independent of the size of the system.  
 We provide numerical evidence showing that the bound is in fact a good estimate even for moderate system sizes, and becomes tighter as the size increases. 

Our study can be extended to a large class of $2$-point correlation functions. 
For instance, for 
 symmetric correlation functions $C_s^A(t) \equiv \langle A(t) A \rangle_\beta + \langle A A(t) \rangle_\beta$, 
 we find that  
  evolution is dominated by a timescale which is at most of 
  order  $t^2 \sim \frac{\langle A^2\rangle_\beta}{\langle [A,H][H,A] \rangle_\beta }$. We argue that this can be interpreted in terms of a fluctuation-dissipation theorem that arises from the unitary dynamics of the system. Finally, we consider the timescales of evolution of the Kubo correlation function that appears in linear response theory~\cite{kubo1957statistical,Khatami13}, which dictates the response of a system at equilibrium to a perturbation in its Hamiltonian.

\

\noindent
\emph{Late-time behaviour ---} 
We now show the rigorous formulation of the late-time factorization of  $2$-point functions.
First, we need the following definition.
\begin{definition}[Clustering of correlations]
\label{def:clustering}
	A state $\rho$ on an Euclidean lattice $\mathbb{Z}^D$ has finite correlation length $\xi >0$ if it holds that
	\begin{equation}
	\max_{X\in M,Y \in N} \frac{\vert \tr{\rho X\otimes Y} -\tr{\rho X} \tr{\rho Y} \vert}{\vert \vert X \vert \vert \vert \vert Y \vert \vert}\le e^{-\frac{\textnormal{dist}(M,N)}{\xi}},
	\end{equation}
	where $M,N$ are regions on the lattice separated by a distance of at least  $\textnormal{dist}(M,N)$, \alv{and $X,Y$ are arbitrary operators with support on each region}.
\end{definition}
This condition is generic of thermal states at finite temperature away from a phase transition. It has been 
proven at least for 1D systems \cite{araki1969gibbs} and arbitrary models above a threshold temperature \cite{kliesch2014locality}.
 In order to prove factorization at late times,
we focus on systems on states that show clustering of correlations, and 
whose Hamiltonians are $k$-local, i.e. which can be written as $H=\sum_j h_j$, where $h_j$ couples at most $k$  closest neighbors.

Given that evolution is unitary and the system is finite-dimensional, limits such as $\lim_{t\rightarrow \infty} C^{AB}(t)$ are not well-defined.  Hence, we consider the 
  late-time behaviour 
  under \emph{infinite-time averages} of the correlation functions
  $\lim_{T\rightarrow \infty} \int_0^T \frac{\text{d}t}{T} C^{AB}(t)$. 
With these considerations, our first main result is the following. 

\begin{theorem}
\label{thm:decoupling}
	Let $H$ be a $k$-local, translation-invariant, non-degenerate Hamiltonian on a $D$-dimensional Euclidean lattice of $N$ sites, and let $[\rho,H]=0$ be an equilibrium ensemble (such as a thermal state) of finite correlation length $\xi>0$. Let $A,B$ be local observables with support on  at most \alv{$\mathcal{O}(N^{\frac{1}{D+1}-\nu})$ sites, with $\nu>0$}. Then
	\begin{align}
	\lim_{T\rightarrow \infty} \int_0^T C^{AB}(t) \frac{\text{d}t}{T}=&\tr{\rho A} \tr{\rho B} \nonumber \\&+\mathcal{O} \left( \xi^{\frac{2 D}{D+1}} \log^2 (N) N^{-\frac{2}{D+1}} \right)
	\end{align}
\end{theorem}
 This guarantees that all operators supported on a region \alv{with size scaling like any function smaller than $\mathcal{O}(N^{\frac{1}{D+1}})$,} satisfy the assumptions of the theorem.
The proof, found in Appendix \ref{app:latetime}, relies on a weak form of the Eigenstate Thermalization Hypothesis (ETH)
shown in \cite{brandao2017quantum}, which is itself based on previous works on large deviation theory for lattice models \cite{anshu2016concentration,mori2016weak}. 
This shows that, in fact, any model obeying the weak ETH and without too many degeneracies will display identical factorization of correlation functions at long times~\cite{biroli2010effect}.

 Note that we assume that the energy spectrum is non-degenerate, which is accurate for systems without non-trivial symmetries or extensive number of conserved quantities. In particular, non-integrable systems usually display Wigner-Dyson statistics in their fine-grained spectrum, which imply level repulsion \cite{d2016quantum}.

This factorization of the correlation function can be thought of as a signature of the emergence of dissipation due to unitary dynamics, since the lack of correlations at different times indicates the loss of information about an initial perturbation of $B$ at time $t=0$, as reflected in the observable $A$ at time $t$~\cite{kubo1957statistical}.

\

\noindent
\emph{Fluctuations around late-time value ---}
For most times, the 2-point correlation function is in fact close to its late-time average, with small fluctuations around the equilibrium value.
 In order to prove this, one needs the extra assumption that the energy gaps are non-degenerate, which is again reasonable in non-integrable systems with connected Hamiltonians \cite{d2016quantum,gogolin2016equilibration},
where it is generally expected to hold as random perturbations are sufficient to lift degeneracies in energy gaps~\cite{Bookcohen2006quantum}.

Let us define $C^{AB}_{\infty}=\lim_{T\rightarrow \infty} \int_0^T \frac{\text{d}t}{T} C^{AB}(t)$, and the average fluctuations around the late-time value as
\begin{align}
\sigma_C^2=\lim_{T \rightarrow \infty} \int_0^T \frac{\text{d}t}{T} \left( C^{AB}(t)-C^{AB}_\infty \right)^2.
\end{align}
The following result puts an upper bound on {average fluctuations.
\begin{theorem}
\label{thm:equilibration}
	Let \alv{$H=\sum_j E_j \ket{j}\bra{j}$} be a Hamiltonian with non-degenerate energy gaps, such that
	\begin{equation}
	E_j-E_k=E_m-E_l \,\, \Leftrightarrow j=m\,\,,k=l,
	\end{equation}
and let $[\rho,H]=0$. It holds that  
\begin{align}\label{eq:boundf}
\sigma_C^2 \le& \, \|A\| \, \|B\| \max_{j \neq k} \{ \left| A_{kj} B_{jk} \right| \}    \tr{ \rho^2 } ,
\end{align}
\alv{where $A_{kj}, B_{jk}$ are matrix elements in the energy eigenbasis,  $A = \sum_{jk} A_{jk} \ket{j}\! \bra{k}$}.
\end{theorem}

The proof can be found in Appendix \ref{app:late-time}. It follows the same steps as the main result in \cite{short2011equilibration}. 
Here, we also find that the purity $\tr{\rho^2}$ of the equilibrium ensemble plays a key role. For a microcanonical ensemble $\tr{(\id / d)^2}=1/d$, so the RHS of Eq. \eqref{eq:boundf} is expected to decay exponentially with system size in most situations of interest. Also, notice that for a thermal state $\tr{\rho_\beta^2}\le 1/Z_\beta$. Moreover, the ETH predicts that $\vert A_{kj} B_{jk}\vert  \sim 1/d$ \cite{deutsch2018eigenstate}.

\

\noindent
\emph{Timescales of equilibration ---} 
Theorems~\ref{thm:decoupling} and~\ref{thm:equilibration} combined imply that correlation functions of the form $\langle A(t) B \rangle_\beta$ are, for most times $t \in [0,\infty]$,
 close to the uncorrelated average $\langle A \rangle_\beta \langle B \rangle_\beta$, for a wide class of  translation-invariant systems.
It is expected that the timescale at which this happens may depend on a number of factors, such as the distance between $A$ and $B$. If the operators are far apart on the lattice the correlations are limited by the Lieb-Robinson bound \cite{lieb1972finite,huang2018lieb}, and timescales associated with ballistic ($\propto N^{1/D}$) or diffusive ($\propto N^{2/D}$) processes may play a role. However, for the autocorrelation function $C^{A}(t) \equiv \langle A(t)A \rangle_\beta$, we can show that  equilibration to the late-time value occurs in a short timescale, independent of system size. 
There may also be further effects at larger timescales, such as the Thouless time~\cite{dymarsky2018bound,schiulaz2018thouless}, and for those effects our result limits their relative size.

%

\alv{Let us define $\rho = \sum_j \rho_{jj} \ket{j}\! \bra{j}$, so that $\rho_{jj}$ and $A_{jk}$ are the matrix elements of $\rho$ and $A$ in the energy basis.} We can then write
\begin{align}
\label{eq:correlationfunction}
\frac{C^{A}(t)}{C^{A}(0)} &= \sum_{j k} \frac{\rho_{jj} |A_{jk}|^2}{C^A(0)} e^{-i(E_j - E_k)t} \equiv \sum_{\alpha} v_\alpha e^{-iG_\alpha t}, 
\end{align}
where we denote \alv{each} pair of levels $\{i,j\}$ by a Greek index, and the corresponding energy gaps by $G_\alpha\equiv E_j-E_k$ \alv{(notice that both $E_j-E_k$ and $E_k-E_j$ appear in the sum)}.   The normalized distribution $v_\alpha \equiv \frac{\rho_{jj} |A_{jk}|^2}{C^A(0)}$ is central to our proofs, since it contains all the relevant information about the state, observable and Hamiltonian, and determines which frequencies contribute to the dynamics of the autocorrelation function.  Based on it, we define the following functions.
\begin{definition}
\label{def:xi}
	Given a normalized distribution $p_\alpha$ over  energy gaps $G_\alpha$, we define $\xi_p(x)$ as the maximum weight that fits an interval of energy gaps with width $x$:
	\begin{align}\label{eq:xidef}
\xi_p(x) \equiv \max_{\lambda} \sum_{\alpha : G_\alpha \in \left[ G_\lambda, G_\lambda + x \right] } p_\alpha.
\end{align}
We also define
\begin{align}
a(\epsilon) \equiv \frac{ \xi_p(\epsilon) }{\epsilon} \sigma_G, \qquad \delta(\epsilon) \equiv \xi_p(\epsilon),
\end{align}
where $\sigma_G = \sqrt{ \sum_\alpha p_\alpha G_\alpha^2 - \left( \sum_\alpha p_\alpha G_\alpha \right)^2 }$ is  the standard deviation of the distribution $p_\alpha$ over the energy gaps $G_\alpha$.  
\end{definition}
 The important point behind these definitions is that, for a sufficiently smooth and unimodal probability distribution, one can find an $\epsilon$ small enough such that $a(\epsilon)\sim\mathcal{O}(1)$ and $\delta(\epsilon)\ll 1$.
In the following theorem, the relevant probability distribution is given by $v_\alpha$. 
Our main result regarding the timescales of correlation functions, proven in Appendix \ref{app:timescaleCorr},  is:
\begin{theorem}
\label{thm:timescalesCorr}
For any Hamiltonian $H$ and state $\rho$ such that $[H,\rho] = 0$, and any observable $A$, the autocorrelation function $C^A(t) = \tr{\rho A(t) A }$ satisfies
\begin{align}
\label{eq:boundCorrAveraged}
\frac{1}{T}\int_0^T \frac{ |C^A(t) - C^A_\infty|^2  }{\left(C^A(0)\right)^2} dt  &\le 3 \pi    \left( \frac{a(\epsilon)}{\sigma_G} \frac{1}{T} + \delta(\epsilon) \right)   
\end{align}
 for any $\epsilon>0$. Here, $a(\epsilon)$ and $\delta(\epsilon)$ are as in Definition~\ref{def:xi} for the normalized distribution $v_\alpha \equiv \frac{\rho_{jj} |A_{jk}|^2}{C^A(0)}$, and 
 $\sigma_G$ is given by
\begin{align}
\sigma_G^2 &= \frac{ 1 }{ C^A(0) } \tr{\rho [A ,H] [H,A ]}  - \frac{\tr{\rho [H,A ] A }^2}{\left(C^A(0)\right)^2}.
\end{align}
\end{theorem}

\noindent Theorem~[\ref{thm:timescalesCorr}] provides an upper bound of $T_{eq} \equiv \frac{3\pi \, a(\epsilon)}{\sigma_G}$ on the timescales under which autocorrelation functions approach their steady state value.  
To see this note that, if for a given $T$ the RHS of Eq. \eqref{eq:boundCorrAveraged} is small, $C^A(t)$ must have spent a significant amount of time  during the interval $[0,T]$ near the late-time value $C^A_\infty$ .

 The crucial point is that   for distributions $v_\alpha$ that are uniformly spread over many values of the gaps $G_\alpha$, one can always find an $\epsilon$ such that $\delta \ll 1$. In that case, the right hand side of Eq.~\eqref{eq:boundCorrAveraged} becomes small on timescales $\mathcal{O}(T_{eq})$. 
As discussed in~\citep{Garcia-PintosPRX2017} and in Appendix \ref{app:aanddelta}, if one further assumes smooth unimodal distributions, typically  one also finds that $a \sim \mathcal{O}(1)$. In that case,  the timescale is governed by $1/\sigma_G$. 
 Given that $\sigma_G$ is a combination of expectation values of local observables, it does not scale with the system of the system.
In fact, a result of~\cite{kim2015slowest} shows that a timescale of order $1/\sigma_G$ provides a lower bound to the timescale of equilibration, which strongly suggests that our upper bound is tight when the conditions of $a\sim \mathcal{O}(1)$ and $\delta \ll 1$ hold.

As a prime example, for local operators in non-integrable lattice models, in which (as per the ETH) $\vert A_{jk}\vert$ are uniformly distributed around a peak at zero energy gap \cite{Beugeling15,Mondaini17}, one should be able to choose $\epsilon$ such that $a \sim \mathcal{O}(1)$ and $\delta \ll 1$. In Fig. \ref{fig:deltaplotexact} we numerically show that this is indeed the case in a non-integrable Ising model.

Theorem \ref{thm:timescalesCorr} does not make assumptions on the specifics of the Hamiltonian, the observable or the state,  making it completely general. 
However, we do not expect the correlation functions to equilibrate well in all cases, as in some scenarios $a(\epsilon)$ and $\delta(\epsilon)$ will be large  no matter what value of $\epsilon$ is chosen, in which case the RHS of Eq.~ \eqref{eq:boundCorrAveraged} may not become small within reasonable timescales.  This can happen, for instance, in models with highly degenerate energy spectrum. 
 To illustrate this, in Appendix \ref{app:simulation} we compute  these parameters in an integrable model, where we see that the gap degeneracies of the model negatively affect the quantities $a(\epsilon)$ and $\delta(\epsilon)$, making the estimated equilibration timescales longer.

\

\noindent
\emph{Symmetric correlation functions ---} 
The previous results can be extended to other correlation functions, such as
\begin{align}
C^A_s(t) \equiv \frac{1}{2} \tr{ \rho \left\{ A , A(t) \right\}}  = \frac{C^A(t) + C^A(t)^*}{2}.
\end{align}
Along the same lines of Theorem~[\ref{thm:timescalesCorr}], in Appendix \ref{app:timescaleCorrSymm} we prove the following.
\begin{theorem}
\label{thm:timescalesCorrSym}
For any Hamiltonian $H$ and state $\rho$ such that $[H,\rho] = 0$, and any observable $A$, the time correlation function $C_s^A(t) = \tr{\rho \{A, A(t)\} }$ satisfies
\begin{align}
\label{eq:boundCorrAveragedSym}
\frac{1}{T}\int_0^T \frac{ |C^A_s(t) - C^A_{s,\infty}|^2  }{(C^A_s(0))^2} dt  &\le 3 \pi    \left( \frac{a(\epsilon)}{\sigma_G} \frac{1}{T} + \delta(\epsilon) \right),   
\end{align}
 for any $\epsilon>0$. Here,  $a(\epsilon)$ and $\delta(\epsilon)$ are as in Definition ~[\ref{def:xi}] for the normalized distribution $v_\alpha^s \equiv \frac{\rho_{jj} + \rho_{kk}}{2} \frac{|A_{jk}|^2}{C^A_s(0)}$, and  
\begin{align}
\sigma_G^2 &= \frac{ 1 }{ C^A_s(0) } \tr{\rho [A_0,H] [H,A_0]}.
\end{align}
\end{theorem}
\noindent Thus an upper bound for the equilibration timescale is
\begin{align}
\label{eq:equilibrationtimescalesymmetric}
T_{eq}  &=   \frac{3 \pi \, a(\epsilon) \, \sqrt{C^A_s(0)}}{\sqrt{ \tr{\rho [A ,H] [H,A ]} }}, 
\end{align}
where again  
 we expect that for small enough $\epsilon$, $a(\epsilon)\sim \mathcal{O}(1)$ and $\delta \ll 1$ for the same reasons as before. 
The denominator in $T_{eq}$ can be seen as an ``acceleration'' of the symmetric autocorrelation function. 
Eq.~\eqref{eq:equilibrationtimescalesymmetric} can in fact be written as
\begin{align}
\label{eq:equilibrationtimescaleaceleration}
T_{eq} =     \frac{3 \pi \, a(\epsilon) \, \sqrt{C^A_s(0)}}{\sqrt{ \left| \frac{d^2 C^A_s(t) }{dt^2} \big\vert_0    \right| }}. 
\end{align}

Such timescale turns out to be similar to that of a short-time analysis. A Taylor expansion gives  
\begin{align}
C^A_s(t) &=C^A_s(0) \left(  1 -    \frac{1}{2 C^A_s(0) } \frac{d^2 C^A_s(t) }{dt^2} \bigg\vert_0   t^2  \right)  +\mathcal{O}(t^3).
\end{align}
For early times, the above expression decays on a timescale $\tau =  \frac{\sqrt{2}}{3 \pi a(\epsilon)} T_{eq}$, 
identical to our upper bound Eq.~\eqref{eq:equilibrationtimescaleaceleration}
up to a prefactor. 

The timescale of Eq. \eqref{eq:equilibrationtimescalesymmetric} suggests an interpretation in terms of an emergent fluctuation-dissipation theorem. Consider \emph{i)} $T_{eq}$ to be the timescale of dissipation of unitary dynamics, meaning that $\langle A(t) A \rangle_\beta \longrightarrow \langle A \rangle_\beta \langle A \rangle_\beta$ occurs, and \emph{ii)} $C^A_s(0)=\tr{\rho A^2}$ as a measure of the  fluctuations of $A$. Then, Eq. \eqref{eq:equilibrationtimescalesymmetric} gives a proportionality relation between the strength of the fluctuations and the timescale of equilibration, in a similar spirit to what was found in \cite{nation2018quantum} using random matrix theory arguments.

\

\noindent
\emph{Linear response and the Kubo correlation function ---}
As a further application of our methods, we study the evolution of a quantum system under a perturbation of its Hamiltonian. Let the system start in a thermal state, such that $\rho  \propto e^{-\beta (H+\lambda A)}$. Subsequently, the Hamiltonian is slightly perturbed by $\lambda A$,  so that the evolved state is $\rho_t = e^{-i t H} \rho e^{i t H}$. 

It was shown by Kubo~\cite{kubo1957statistical} that, to leading order in $\lambda$, the expectation value of $A$ satisfies $\tr{\rho A(t)}~=~C_{\text{Kubo}}(t)\tr{\rho A}$, where
for thermal initial  states $\rho$  the Kubo correlation function can be written as
\begin{align}
 C_{\text{Kubo}}(t) \propto \sum_{j\neq k}\frac{e^{-\beta E_k}-e^{-\beta E_j}}{E_j-E_k} \vert A_{jk}\vert^2 e^{i t(E_j-E_k)}.
\end{align}
Equilibration of $\tr{\rho A(t)}$ is then equivalent to equilibration of the function $\Ckubo(t)$, for which we prove in Appendix \ref{app:timescaleKubo}
that the following holds.
\begin{theorem}
\label{thm:timescalesKubo}
For any Hamiltonian $H$, thermal state $\rho  \propto e^{-\beta (H+\lambda A)}$, and any observable $A$, the Kubo correlation function $C_{\textnormal{Kubo}}$ satisfies
\begin{align}
\label{eq:boundKubo}
\frac{1}{T}\int_0^T \frac{ |C_{\textnormal{Kubo}}(t) - C_{\textnormal{Kubo},\infty} |^2  }{C_{\textnormal{Kubo}}(0)^2} dt  &\le 3 \pi    \left( \frac{a(\epsilon)}{\sigma_G} \frac{1}{T} + \delta(\epsilon) \right),   
\end{align}
 for any $\epsilon>0$. Here,  $a(\epsilon)$ and $\delta(\epsilon)$ are as in Definition ~[\ref{def:xi}] for the normalized distribution $w_\alpha \equiv \frac{e^{-\beta E_k}-e^{-\beta E_j}}{E_j-E_k} \frac{ \vert A_{jk}\vert^2 }{C_{\textnormal{Kubo}}(0)} $, and 
\begin{align}
\sigma_G^2 &= \frac{ 1 }{ \Ckubo(0) } \tr{ [A,\rho] [A,H]}.
\end{align}
\end{theorem}
This again implies an upper bound $T_{eq} =     \frac{3 \pi \, a(\epsilon)}{\sigma_G}$ on the equilibration timescale of $\Ckubo$, 
and therefore on the time to return to thermal equilibrium after a perturbation of the system Hamiltonian by $A$. 
\alv{ The distribution $w_\alpha$ plays the same role as $v_\alpha$ and $v^s_\alpha$ before. If $w_\alpha$ is smoothly distributed and unimodal (which we expect for local observables in non-integrable models)  then $a \sim \mathcal{O}(1)$ and $\delta(\epsilon)\ll 1$ holds (see Appendix \ref{app:aanddelta})}.

\

\noindent
\emph{Simulations ---}
We test Theorem~[\ref{thm:timescalesCorr}] in a spin model governed by the Hamiltonian
\begin{equation} \label{eq:hamiltonian}
H = \sum_{j=1}^L\left(\gamma \sigma_j^X + \lambda \sigma_j^Z   \right) + J\sum_{j=1}^{L-1} \sigma_j^Z\sigma_{j+1}^Z+ \alpha\sum_{j=1}^{L-2} \sigma_j^Z\sigma_{j+2}^Z,
\end{equation}
where $\sigma_j^Z$ and $\sigma_j^X$ are the Pauli spin operators along $Z$ and $X$ directions for spin $j$, and
we take open boundary conditions.
The field and interaction coefficients $(\gamma,\lambda, J, \alpha )$ characterize the model. 
 We focus on a case corresponding to a system satisfying ETH by choosing $(\gamma,\lambda, J, \alpha ) = (0.8,0.5,1,1)$~\cite{kaneko2018work},
  and study the autocorrelation functions of the observable $A = \sigma_{\frac{L}{2}}^x.$ For simplicity we set $\beta = 1$ in our numerics, though no significant changes were observed for $\beta \in [0.1,5]$.
Figure~\ref{fig:deltaplotexact} depicts the functions $a(\epsilon)$ and $\delta(\epsilon)$ that appear in Theorem~\ref{thm:timescalesCorr},  confirming that there exist regions of $\epsilon$ such that $\delta \ll 1$, ensuring equilibration occurs, and $a \sim 0.4$. Importantly, this is increasingly the case as the size of the system grows.

Figure~\ref{fig:boundplot} compares the two sides in bound~\eqref{eq:boundCorrAveraged}, 
showing that dynamics obtained from the upper bound differs from the actual dynamics by roughly an order of magnitude.
Thus, the general, model-independent bounds obtained from Theorem~\ref{thm:timescalesCorr} provide remarkably good estimates of the actual (simulated) dynamics.
Note that the estimate becomes increasingly better as the size of the system increases. 
This discrepancy could, however, be a finite-size effect, which is also suggested by the lower bound obtained in~\cite{kim2015slowest}.
Details of the simulations can be found in \ref{app:simulation}.

\begin{figure}[h]
	\centering
	\includegraphics[width = 0.9\linewidth]{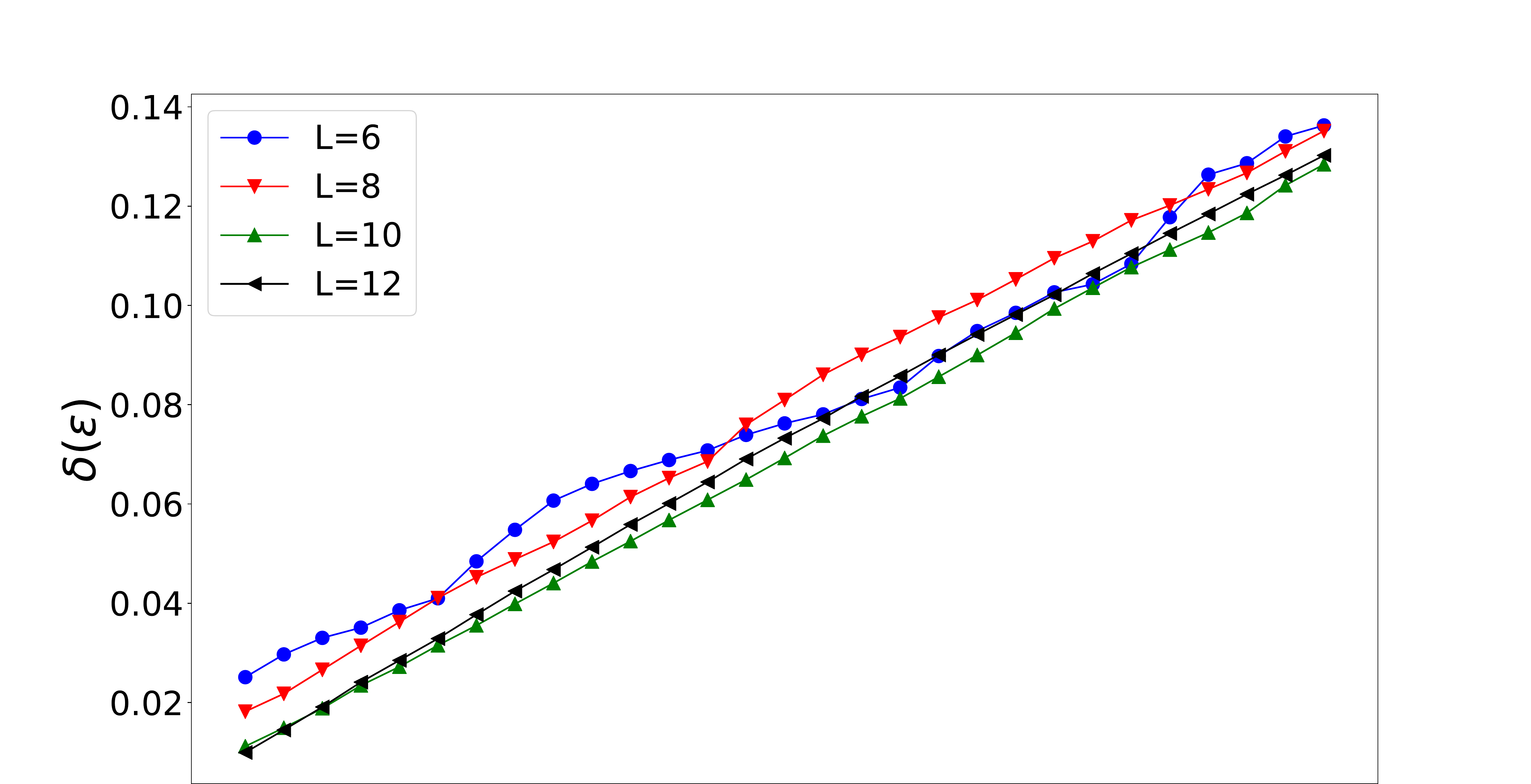}
		\includegraphics[width = 0.9\linewidth]{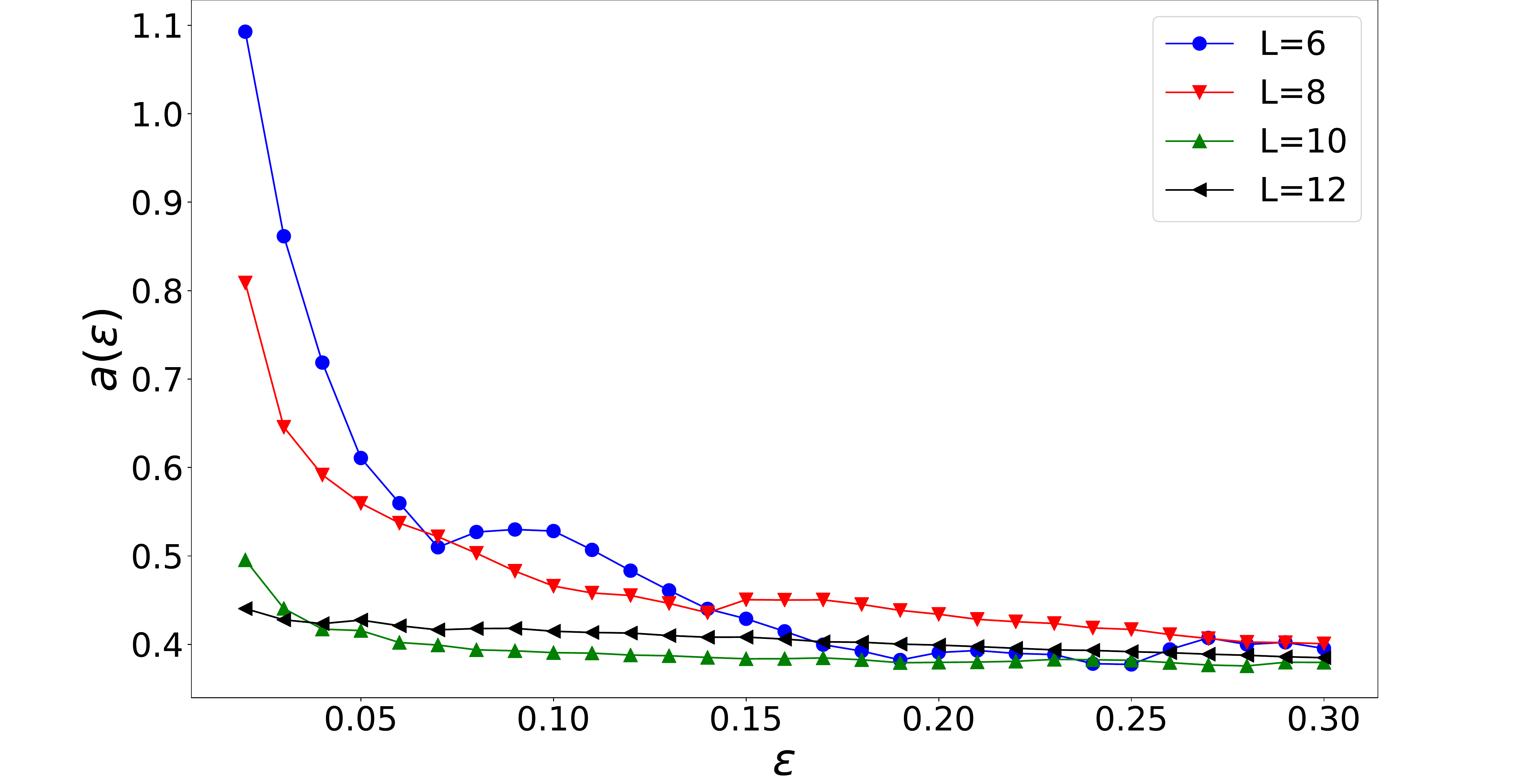}
	\caption{Plots of $\delta(\epsilon)$ (top) and $a(\epsilon)$ (bottom) for distribution  $v_\alpha \equiv \frac{\rho_{jj} |A_{jk}|^2}{C^A(0)}$  in Theorem~\ref{thm:timescalesCorr}, obtained by exact diagonalization and a Monte Carlo approximation.
	 The plots 
	 were generated with $10,000$ sampled frequency intervals. 
Small values of $\delta$ imply equilibration occurs for long enough times, while the value of $a$ controls the prefactor in the equilibration
 timescale 
$T_{eq} \equiv \frac{3\pi \, a(\epsilon)}{\sigma_G}$ derived from Eq.~\eqref{eq:boundCorrAveraged}.
 For small $\epsilon$ one can satisfy both $\delta \ll 1$ and $a \sim \mathcal{O}(1)$, and this becomes increasingly so for larger system sizes.} 
	\label{fig:deltaplotexact}
	\label{fig:aplotexact}
\end{figure}
\begin{figure}[ht]
	\centering
	\includegraphics[width=0.45\textwidth]{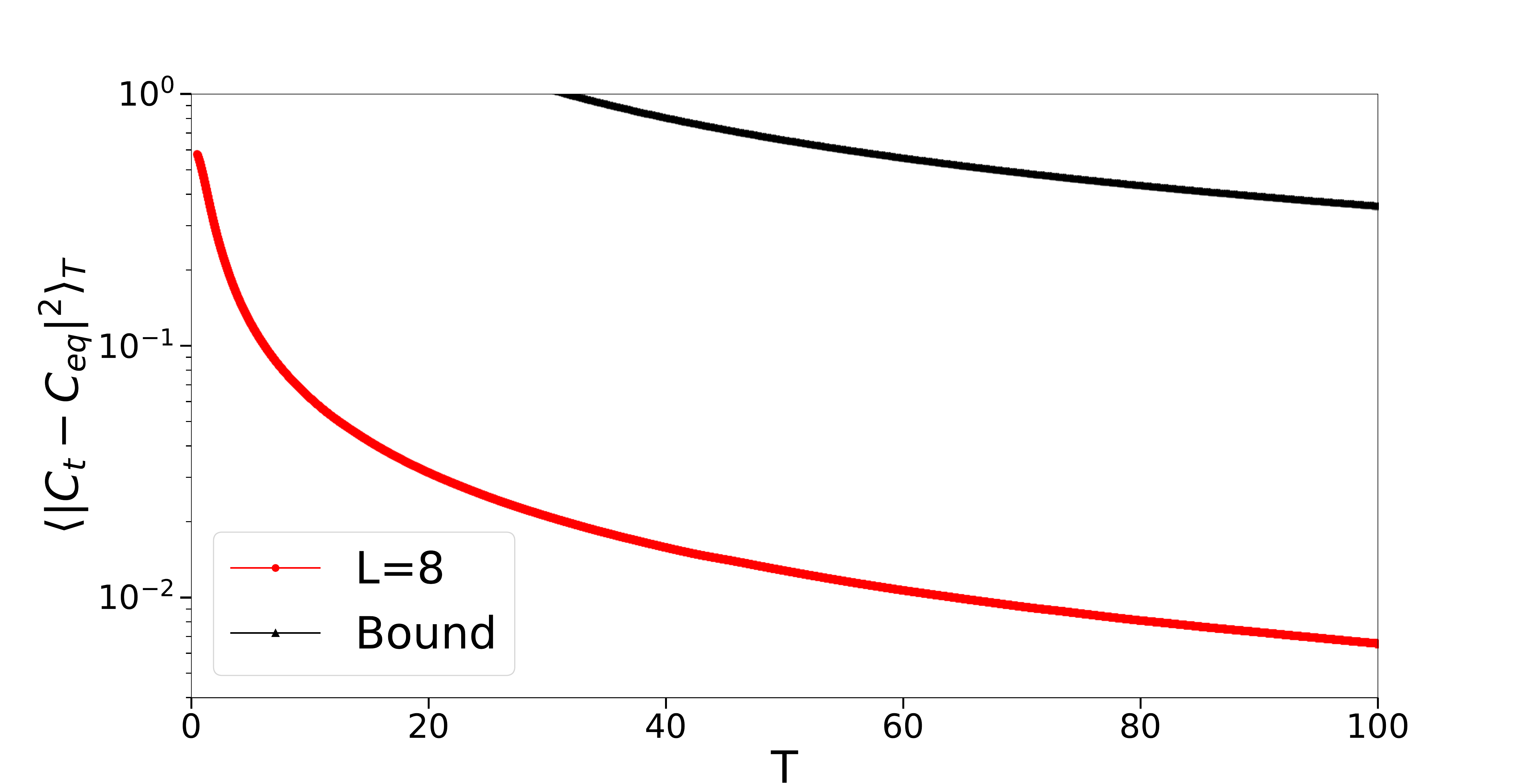} 
	\includegraphics[width=0.45\textwidth]{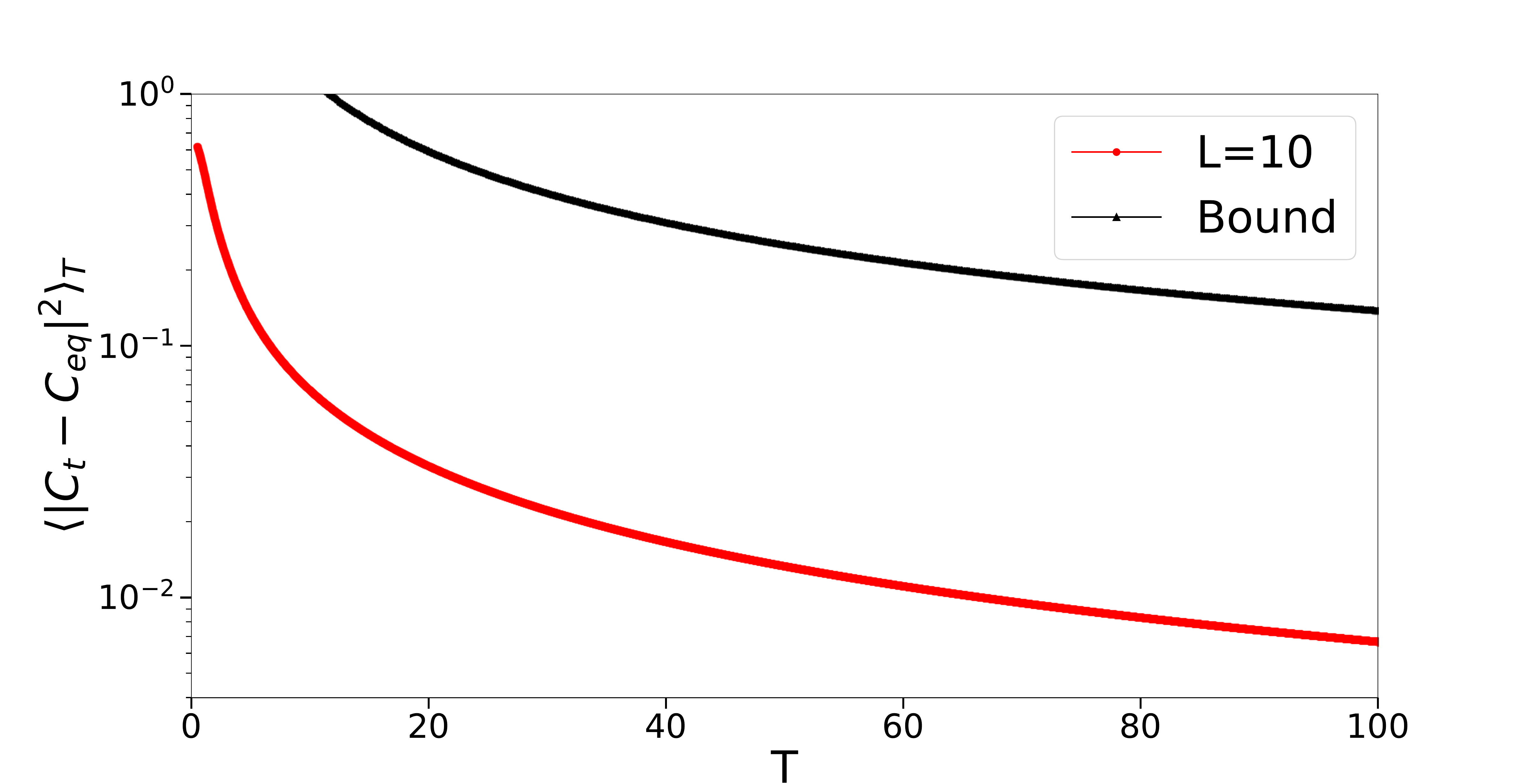} \includegraphics[width=0.45\textwidth]{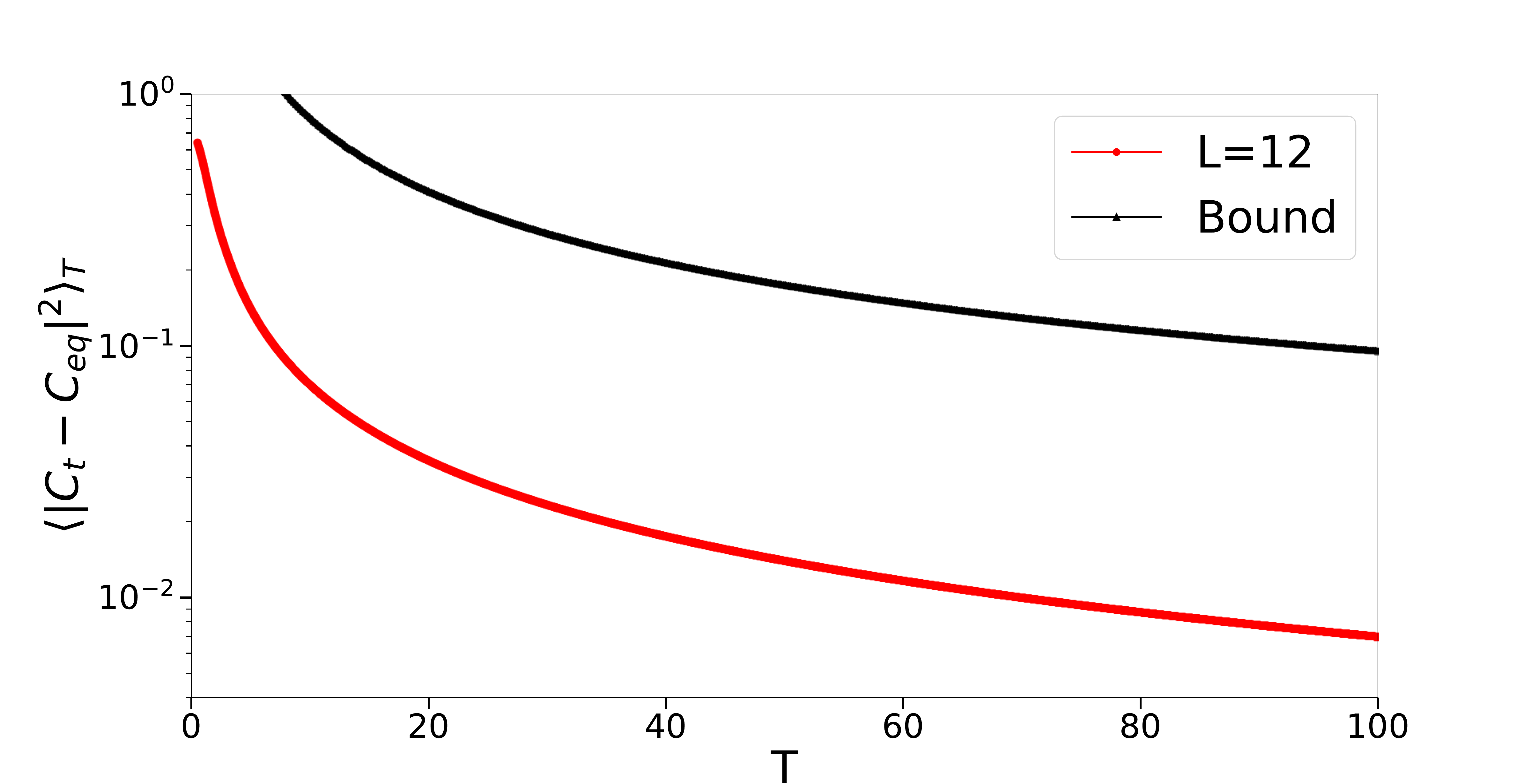} \caption{
		\label{fig:figure4SM} 
		Comparison of the upper bound in Eq.~\ref{eq:boundCorrAveraged}  (RHS) with the simulated evolution of the time-averaged correlation function (LHS) as a function of time, for increasing number of spins $L$. The evolution obtained from the upper bound approaches the exact dynamics of the system for larger system size. 
	}
	\label{fig:boundplot}
\end{figure}

\

\noindent\emph{Discussion ---} 
We derived analytic results on the dynamical behavior of $2$-point correlation functions in quantum systems. These include conditions that imply that time-correlation functions factorize for long times, as well as easy-to-estimate upper bounds on the timescales under which such process occurs
 which hold regardless of details of the model under consideration.
Remarkably, 
our numerical findings show that the
derived upper bounds can correctly estimate the actual dynamics of the system to within an order of magnitude, and
become increasingly better
  estimates as the size of the system increases. 



We used techniques previously applied in the context 
of equilibration of quenched quantum systems \cite{ShortFarrelly11,Malabarba14,Garcia-PintosPRX2017}, for which finding rigorous estimates on the timescales is a largely open problem \cite{Wilming17,de2018equilibration,dymarsky2018mechanism,Reimann_2019}. This connection is not surprising, specially considering that previous works \cite{Srednicki99,richter2018impact,richter2019relation} have argued that in some situations one can approximate the out of equilibrium dynamics with the autocorrelation functions covered here. 

Given the importance of time-correlation functions in the analysis of a wide range of problems in many-body physics --for instance, in transport phenomena-- we anticipate that our results will be useful in the description of closed system dynamics, whose study has surged in recent times due to enormous experimental advances in settings such as cold atoms or ion traps \cite{bloch2008many,schneider2012experimental}.

\

\noindent
\begin{acknowledgments}
The authors acknowledge useful discussions with Anurag Anshu, Beni Yoshida, Charlie Nation and Jens Eisert. We are also thankful to Yichen Huang for pointing out an improvement to the error term in Theorem~\ref{thm:decoupling} from a previous version (see \cite{HuangArxiv2019}). This research was supported in part by the Perimeter Institute for Theoretical Physics. Research at
Perimeter Institute is supported by the Government of
Canada through the Department of Innovation, Science
and Economic Development and by the Province of Ontario
through the Ministry of Research, Innovation and
Science. This research was also supported by NSERC and 
enabled in part by support provided by (SHARCNET) (www.sharcnet.ca) and Compute/Calcul Canada (www.computecanada.ca).
LPGP acknowledges support from the John Templeton Foundation, UMass Boston Project No. P20150000029279,
DOE Grant No. DE-SC0019515, AFOSR MURI project "Scalable Certification of Quantum Computing Devices and Networks", DoE ASCR Quantum Testbed Pathfinder program (award No. DE-SC0019040), DoE BES QIS program (award No. DE-SC0019449), DoE ASCR FAR-QC (award No. DE-SC0020312), NSF PFCQC program, AFOSR, ARO MURI, ARL CDQI, and NSF PFC at JQI.
\end{acknowledgments}
\bibliography{references}

\widetext
\appendix	
\section{Late time behaviour of two-point functions}
\subsection{Proof of late-time equilibration}\label{app:latetime}

First, we state the key result from \cite{brandao2017quantum} that we use.  It says that expectation values for single eigenstates, of the form $\bra{E_k} A \ket{E_k}$, are close to the ensemble average $\langle A \rangle_\beta$ with \emph{very high probability}. In contrast, the strong form of the ETH states that the above happens for \emph{all} eigenstates within an energy window. We reproduce the proof of \cite{brandao2017quantum} (which itself builds on \cite{mori2016weak}), with the difference that our version holds for lattices of dimension larger than $1$.

\begin{lemma} [Proposition 7 \cite{brandao2017quantum}]\label{th:weakETH}
	Let $H$ be a translation-invariant, non-degenerate Hamiltonian with $N$ sites on a $D$-dimensional lattice, $\rho$ an equilibrium ensemble $[\rho,H]=0$ with finite correlation length $\xi$, and $A$ some observable with support on a connected region of at most $\mathcal{O}(N^{1/(D+1)-\nu})$ sites, with $\nu>0$. Then, for any $\delta > 0$,
	\begin{equation}
	\text{Pr}_{\ket{E_k}\in \rho} (\vert \bra{E_k} A \ket{E_k}-\tr{\rho A}\vert \ge \delta )\le \exp{(-c \delta N^{\frac{1}{D+1}} \xi^{-\frac{D}{D+1}})},
	\end{equation}
	where $c>0$ is a constant, and $\ket{E_k}\in \rho$  indicates that the eigenstates are sampled from the equilibrium distribution $\rho$.
\end{lemma}
\begin{proof}
	We show a bound $\text{Pr}_{\ket{E_k}\in \rho} ( \bra{E_k} A \ket{E_k}-\tr{\rho A} \ge \delta )$, and $\text{Pr}_{\ket{E_k}\in \rho} ( \tr{\rho A}- \bra{E_k} A \ket{E_k} \ge \delta )$ will follow analogous steps. Notice that since the Hamiltonian is translation-invariant and non-degenerate we can write $\bra{E_k} A \ket{E_k}=\bra{E_k} \bar{A}/N  \ket{E_k}$,
	where $\bar{A}=\sum_i A_i$ is the extensive observable built out of translations of $A$. Define $\bar{A}'=\bar{A}-\langle \bar{A} \rangle_\beta$. Then, using Markov's inequality and $e^{\bra{\psi} A \ket{\psi}}\le \bra{\psi} e^A \ket{\psi}$  we can write
	\begin{align}
	\text{Pr}_{\ket{E_k}\in \rho} ( \bra{E_k} A \ket{E_k}-\tr{\rho A} \ge \delta ) &= \text{Pr}_{\ket{E_k}\in \rho}(e^{\lambda \bra{E_k} \bar{A}' \ket{E_k} }\le e^{\lambda \delta N}) \\
	& \le e^{- \lambda \delta N } \sum_k \rho_{kk} e^{\lambda \bra{E_k} \bar{A}' \ket{E_k} } \\ & \le e^{- \lambda \delta N } \tr{\rho e^{\lambda \bar{A}'}}.\label{eq:prbound}
	\end{align}
	Now let us write the spectral decomposition $e^{\lambda \bar{A}'}=\sum_l e^{\lambda a_l} \Pi_{l}$, with $\Pi_l$ the projectors into the subspace with eigenvalue $e^{\lambda a_l}$, and write the average as
	\begin{equation}\label{eq:decomp}
	\tr{\rho e^{\lambda \bar{A}'}}=\sum_{a_l \le \delta N /2 }  e^{\lambda a_l} \tr{\rho \Pi_{l}} + \sum_{a_l > \delta N /2 } e^{\lambda a_l} \tr{\rho \Pi_{l}}.
	\end{equation}
	The first term is upper bounded by $e^{\lambda \delta N/2}$. For the second, we write
	\begin{align}
	\sum_{a_l > \delta N /2 } e^{\lambda a_l} \tr{\rho \Pi_{l}} =\sum_j \sum_{\delta N/2+(j+1) \ge a_l \ge \delta N/2+j} e^{\lambda a_l} \tr{\rho \Pi_{l}} \le \sum_j e^{\lambda (\delta N  /2+j+1)} \tr{\rho \Pi_{\ge \delta N/2+j}},
	\end{align}
	where $\Pi_{\ge x}$ denotes the projector on the subspace with $a_l \ge x$. The main result of \cite{anshu2016concentration} states that
	\begin{equation}
	\tr{\rho \Pi_{\ge \delta N/2+j}} \le e^{D k N^\alpha /\xi} \text{exp}\left(- \frac{c'}{D \xi}\left(N (\delta/2+j/N)^2 \xi\right)^{1/D+1} \right).
	\end{equation}
	Here, $k N^\alpha$ is the number of sites on which an individual local term of observable $\bar{A}$ has support.
	Since $j$ is at most $\mathcal{O}(N)$, we can choose some $\lambda =\mathcal{O}\left((N \xi)^{-\frac{D}{D+1}} \right)$ such that, for some constant $c_1>0$,
	\begin{equation}
	\sum_{a_l > \delta N /2 } e^{\lambda a_l} \tr{\rho \Pi_{l}} \le \text{poly}(N) e^{D k N^\alpha /\xi} \text{exp}\left(-c_1 N^{\frac{1}{D+1}} \xi^{-\frac{D}{D+1}}\right). 
	\end{equation}
	This way, the dominant contribution of Eq. \eqref{eq:decomp} is the first term. Plugging the bounds back in Eq. \eqref{eq:prbound} results in the following, for some constant $c>0$ and large enough $N$,
	\begin{align}\label{eq:deltab}
	\text{Pr}_{\ket{E_k}\in \rho} ( \bra{E_k} A \ket{E_k}-\tr{\rho A} \ge \delta ) &\le e^{-\lambda \delta N/2}\left(1 + e^{-\lambda \delta N/2} \text{poly}(N) e^{D k N^\alpha /\xi} \text{exp}\left(-c_1 N^{\frac{1}{D+1}} \xi^{-\frac{D}{D+1}}\right)\right) \\& \le \text{exp}\left(-c \delta N^{\frac{1}{D+1}} \xi^{-\frac{D}{D+1}}\right),
	\end{align}
	where the last line follows from the fact that the second term in the RHS of Eq.\eqref{eq:deltab} is subleading (much smaller than $1$) as long as $\alpha<\frac{1}{D+1}$. \alv{This sets the support of of the observable $A$ ot be within a region of size at most $\mathcal{O}(N^{1/(D+1)-\nu})$ for any $\nu>0$}. 
\end{proof}
With it, we are now ready to prove the result on late-time factorization of correlation functions.
\begin{customthm}{1}
	Let $H$ be a $k$-local, translation-invariant and non-degenerate Hamiltonian on a $D$-dimensional Euclidean lattice of $N$ sites, and let $[\rho,H]=0$ be an equilibrium ensemble (such as a thermal state) of finite correlation length $\xi>0$. Let $A,B$ be local observables with support on at most \alv{$\mathcal{O}(N^{\frac{1}{D+1}-\nu})$ sites, with $\nu>0$}. Then
	\begin{equation}
	\lim_{T\rightarrow \infty} \int_0^T \tr{\rho A(t) B} \frac{\text{d}t}{T}=\tr{\rho A} \tr{\rho B}+\mathcal{O}(\xi^{ \frac{2 D}{D+1}} \log^2 (N) N^{-\frac{2}{D+1}})
	\end{equation}
\end{customthm}
\begin{proof}
	First let us write
	\begin{align}
	\rho=\sum_k \rho_{k k} \ket{E_k}\bra{E_k} \,\,,\,\, A=\sum_{k j} A_{k j} \ket{E_k}\bra{E_j}\,\, , \,\, B=\sum_{k j} B_{k j} \ket{E_k}\bra{E_j},
	\end{align}
	from which we have
	\begin{equation}
	\tr{\rho A(t) B}=\sum_{k j} A_{k j} B_{j k} \rho_{kk} e^{i t (E_j-E_k)}. 
	\end{equation}
	Since the Hamiltonian is non-degenerate by assumption, it holds that 
	\begin{equation}
	\lim_{T\rightarrow \infty} \int_0^T e^{i t (E_j-E_k)} \frac{d t}{T}=\delta_{k,j},
	\end{equation}
	so that the limit becomes
	\begin{equation}
	\lim_{T\rightarrow \infty} \int_0^T \tr{\rho A(t) B} \frac{dt}{T}=\sum_{k} A_{k k} B_{k k} \rho_{kk}.
	\end{equation}
	Now let us define $A_{k k}-\tr{\rho A}=\Delta_{k,A}$ and $B_{k k}-\tr{\rho B}=\Delta_{k,B}$ , so that~\cite{HuangArxiv2019} 
	\begin{align}\label{eq:errorterms}
	\sum_{k} A_{k k} B_{k k} \rho_{kk}&=\tr{\rho A}\tr{\rho B}+\sum_k \rho_{k k}(\tr{\rho A}\Delta_{k,B}+\tr{\rho B} \Delta_{k,A}+\Delta_{k,A}\Delta_{k,B})\\ & =\tr{\rho A}\tr{\rho B}+\sum_k \rho_{k k}\Delta_{k,A}\Delta_{k,B}
	\end{align}
	Let us define $\Delta \equiv K \log N/ N^{\frac{1}{D+1}}$, and split the sum over energies of the error term as
	\begin{equation}\label{eq:sumk1}
	\sum_k \rho_{k k}\Delta_{k,A}\Delta_{k,B}=\sum_{k \in \mathcal{S}} \rho_{k k}\Delta_{k,A}\Delta_{k,B}+\sum_{k \notin \mathcal{S}} \rho_{k k}\Delta_{k,A}\Delta_{k,B}.
	\end{equation}
	where $\mathcal{S}\equiv \{k: |\Delta_{k,A}|,|\Delta_{k,B}|\le\Delta \}$ (that is, the set of $k$ for which both errors are small).
	Notice that the first term is smaller than $\Delta^2$ by definition. On the other hand, the second term can be bounded as
	\begin{align}
	\sum_{k \notin \mathcal{S}} \rho_{k k}\Delta_{k,A}\Delta_{k,B} &\le \max_{k'} \Delta_{{k'},A} \Delta_{{k'},B} \sum_{k \notin \mathcal{S}} \rho_{k k} \nonumber \\
	&\le \max_{k'} \Delta_{{k'},A} \Delta_{{k'},B} \left(\sum_{\vert \Delta_{k,A} \vert \ge \Delta} \rho_{k k}+\sum_{\vert \Delta_{k,B} \vert \ge \Delta} \rho_{k k} \right )\nonumber
	\\
	&\le 2\max_{k'} \Delta_{{k'},A}\Delta_{{k'},B} \exp{(-c \Delta N^{\frac{1}{D+1}} \xi^{-\frac{D}{D+1}})} \nonumber  \\
	&\le 2 \vert\vert A \vert \vert \vert \vert B \vert \vert \left( \frac{1}{N} \right)^{c K\xi^{-\frac{D}{D+1}}}.
	\end{align}
	The third line follows from Lemma \ref{th:weakETH}, and the fourth from $|\Delta_{k,A}|\le \vert \vert A \vert \vert,|\Delta_{k,B}|\le \vert \vert B \vert \vert$. The constant $K$ is arbitrary, so we can choose it such that $c K\xi^{-\frac{D}{D+1}} = \frac{2}{D+1}$. In that case the dominant contribution to Eq.~\eqref{eq:sumk1} is that of the first term, and hence $\sum_k \rho_{k k}\Delta_{k,A} \Delta_{k,B} =\mathcal{O}(\Delta^2)$, so that
	\begin{equation}
	\sum_{k} A_{k k} B_{k k} \rho_{kk}=\tr{\rho A}\tr{\rho B}+\mathcal{O}(\Delta ^2),
	\end{equation}
	completing the proof.
\end{proof}

\subsection{Proof of fluctuations around late-time value}
\label{app:late-time}

\begin{customthm}{2}
	\label{thm:equilibration}
	Let \alv{$H=\sum_i E_i \ket{i}\bra{i}$} be a Hamiltonian with non-degenerate energy gaps, such that
	\begin{equation}
	E_j-E_k=E_m-E_l \,\, \Leftrightarrow j=m\,\,,k=l,
	\end{equation}
	and let $[\rho,H]=0$. It holds that  
	\begin{align}\label{eq:boundf}
	\sigma_C^2 \le& \, \|A\| \, \|B\| \max_{j \neq k} \{ \left| A_{kj} B_{jk} \right| \}    \tr{ \rho^2 } ,
	\end{align}
	\alv{where $A_{kj}, B_{jk}$ are off-diagonal matrix elements in the energy eigenbasis,  $A = \sum_{jk} A_{jk} \ket{j}\! \bra{k}$}.
\end{customthm}

\begin{proof}
	Let us expand in the energy eigenbasis.
	\begin{align}
	\sigma_C^2 &=\lim_{T \rightarrow \infty}\int_0^T \frac{\text{d}t}{T} (C(t)^{AB}-C_\infty^{AB})^2 \\ 
	&= \sum_{j\neq k} \sum_{l\neq m} \rho_{jj} \rho_{ll} A_{jk}B_{kj} A_{lm} B_{ml} \int_0^T \frac{\text{d}t}{T} e^{-i t (E_j-E_k+E_l-E_m)}
	\\ &= \sum_{j\neq k} \rho_{jj} \rho_{kk} A_{jk}A_{kj} B_{jk}B_{kj}
	\\ &\le \sum_{j\neq k} \rho_{jj} \rho_{kk} \left| A_{jk} A_{kj} B_{jk}B_{kj} \right|
	\\ &\le  \max_{j \neq k} \{ \left| A_{kj} B_{jk} \right| \}\sum_{j\neq k} \rho_{jj} \rho_{kk} \left | A_{jk}B_{kj} \right|
	\end{align}
	In the second to the third line we used the assumption of non-degenerate energy gaps. 
	We now use the Cauchy-Schwarz inequality twice.
	\begin{align}
	\sigma_C^2 &\le  \max_{j \neq k} \{ \left| A_{kj} B_{jk} \right| \}  \sqrt{ \sum_{j\neq k} \rho_{jj} \rho_{kk} \left | A_{jk} \right|^2 } \sqrt{ \sum_{j\neq k} \rho_{jj} \rho_{kk} \left | B_{kj} \right|^2 }
	\\ &\le  \max_{j \neq k} \{ \left| A_{kj} B_{jk} \right| \}  \sqrt{ \sum_{j,k} \rho_{jj} \rho_{kk} \left | A_{jk} \right|^2 } \sqrt{ \sum_{j,k} \rho_{jj} \rho_{kk} \left | B_{kj} \right|^2 }
	\\ &=   \max_{j \neq k} \{ \left| A_{kj} B_{jk} \right| \}  \sqrt{ \tr{\rho A \rho A} } \sqrt{ \tr{\rho B \rho B} }
	\\& \le  \max_{j \neq k} \{ \left| A_{kj} B_{jk} \right| \}  \sqrt{  \tr{A^2 \rho^2 } }  \sqrt{ \tr{B^2 \rho^2 } }
	\\& \le  \|A\| \, \|B\| \max_{j \neq k} \{ \left| A_{kj} B_{jk} \right| \}    \tr{ \rho^2 }.  
	\end{align}
	The last inequality follows from the fact that for positive operators $\tr{PQ} \le \vert\vert P \vert \vert \tr{Q}$. 
	
\end{proof}

\section{Dynamics of two-time correlation functions}

\subsection{Preliminaries}
\label{app:Prelim}
Here we prove an intermediate lemma required for Appendix \ref{app:timescaleCorr}. This is similar to results found in ~\cite{Malabarba14,Garcia-PintosPRX2017}, improving on them by a numerical factor. Let us first recall the following definition from the main text. 

\begin{customdef}{2}
	\label{def:xi}
	Given a normalized distribution $p_\alpha$ over $G_\alpha$, we define $\xi_p(x)$ as the maximum weight that fits an interval of energy gaps with width $x$:
	\begin{align}
	\xi_p(x) \equiv \max_{\lambda} \sum_{\alpha : G_\alpha \in \left[ G_\lambda, G_\lambda + x \right] } p_\alpha.
	\end{align}
\end{customdef}   
With it, we have a general upper bound that will be central to the later proofs.

\begin{lemma}
	\label{lem:avgbound}
	Let $f(t)$ be a positive function of the form 
	\begin{equation}
	f(t) = \sum_{\alpha,\beta}p_\alpha p_\beta e^{i\left(G_{\alpha}- G_{\beta} \right)t},
	\end{equation}
	such that the $p_\alpha$  form a discrete probability distribution $\sum_\alpha p_\alpha = 1$. The uniform average of such a function is upper bounded by, 
	\begin{equation}
	\langle f(t) \rangle_T \equiv \int_0^T \frac{\text{d}t}{T} f(t)  \leq 3 \pi \xi_p\left(\frac{1}{T}\right), 
	\end{equation}
	where $\xi_p(x)$ is as in Definition \ref{def:xi}. 
\end{lemma}
\begin{proof}
	Let us first define the uniform and Gaussian probability density functions, as 
	\begin{equation}
	\label{eq:uniform}
	p_T(t) =  \left\{
	\begin{array}{cl}
	\frac{1}{T} & t\in [0,T],  \\
	0 &\text{otherwise}.
	\end{array}\right.
	\end{equation} 
	and
	\begin{equation} \label{eq:gaussian}
	p_G(t)  = 	\frac{1}{\sqrt{2\pi \alpha^2T^2}} e^{-\frac{(t-T/2)^2}{2\alpha^2T^2}}, \enspace t\in R,
	\end{equation}
	respectively, 
	where the mean of the Gaussian is written as $\mu = \frac{T}{2}$ and the standard deviation as $\sigma = \alpha T$. The parameter $\alpha>0$ is free for us to choose.
	We will first show that for any positive function $f(t)  \geq 0$, the uniform average  of a function on the interval $t\in [0,T]$ can be bounded tightly by the Gaussian average by, 
	\begin{equation}\label{eq:distbound}
	\langle  f(t) \rangle_T \leq \gamma \langle f(t) \rangle_{G_T},
	\end{equation}
	where $\gamma = \sqrt{2\pi} \alpha e^{\frac{1}{8\alpha^2}}$, and we denote the uniform and Gaussian averages over distributions $p_T(t)$ and $p_G(t)$ as $\langle  f(t) \rangle_T$ and $\langle f(t) \rangle_{G_T}$ respectively. 
	
	To do this, note that due to the choice of $\mu = \frac{T}{2}$ the two distributions have identical means. To find the smallest possible $\gamma$ such that \eqref{eq:distbound} holds we set that $p_T(t=T) = \gamma p_{G}(t=T)$. 
	This then implies,
	\begin{align}
	\frac{1}{T} = \gamma \frac{1}{\sqrt{2 \pi} \alpha T} e^{-\frac{1}{8 \alpha^2}} \implies \gamma = \sqrt{2\pi} \alpha e^{\frac{1}{8\alpha^2}},
	\end{align}
	which guarantees that $p_T(t) \le \gamma p_G(t)\,\,\forall t$. Since over the interval $[0,T]$ $p_T(t)$ is constant and $p_G(t)\ge p_G(T)$,it follows that
	\begin{equation}
	\langle  f(t) \rangle_T \leq \gamma \langle f(t) \rangle_{G_T}.
	\end{equation}
	With this we can now bound the uniform average of $f(t)$ as
	\begin{equation}
	\langle f(t) \rangle_T = \frac{1}{T} \int_{0}^T \sum_{\alpha, \beta}p_{\alpha} p_{\beta} e^{i\left(G_{\alpha}- G_{\beta} \right)t}dt \leq   \frac{\gamma}{\sqrt{2\pi \alpha^2T^2}} \int_{-\infty}^\infty \sum_{\alpha, \beta}p_{\alpha} p_{\beta} e^{i\left(G_{\alpha}- G_{\beta} \right)t} e^{-\frac{(t-T/2)^2}{2\alpha^2T^2}} dt.
	\end{equation}
	Interchanging the integral with the sum and then integrating gives the characteristic function of a Gaussian distribution. This allows us to write, 
	\begin{equation} \label{eq:gaussform}
	\langle f(t) \rangle_T \leq \gamma \sum_{\alpha\beta} p_\alpha p_\beta e^{ - \frac{\sigma^2 \Delta G^2}{2}},
	\end{equation}
	where $\Delta G = G_\alpha-G_\beta$. To further bound Eq. \eqref{eq:gaussform} we may introduce the function, 
	\begin{equation}
	g(x) = \left\{
	\begin{array}{cl}
	1 & \text{ if } x\in[0,1),  \\
	0 &\text{otherwise}.
	\end{array}\right.
	\end{equation}
	This function satisfies
	\begin{equation}
	r^{-x^2} \leq \sum_{n=0}^{\infty} r^{-n^2} g(|x|-n)
	\end{equation}
	where $r > 1$. Applying this bound to Eq. \eqref{eq:gaussform} with $r =  e^{\frac{\alpha}{2}}$ and $x = \Delta G^2 T^2$, it follows that 
	\begin{equation} \label{eq:gstep}
	\langle f(t) \rangle_T \leq \gamma \sum_{n=0}^{\infty} r^{-n^2} \sum_{\alpha,\beta} p_\alpha p_\beta g\left(\Delta G^2 T^2-n\right).
	\end{equation}
	The presence of $g(x)$ allows to restrict the summation in Eq.~\eqref{eq:gstep}. With the condition $\left( \Delta G^2 T^2-n \right) \in [0,1)$, the two following intervals can be identified for the variable $G_\beta$ given a particular value of $G_\alpha$,
	\begin{align}
	G_\beta \in I_+ =  \left[ G_\alpha+ \frac{\sqrt{n}}{T}, G_\alpha+\frac{\sqrt{n+1}}{T} \right), \\ G_\beta \in I_- = \left( G_\alpha -  \frac{\sqrt{n+1}}{T}, G_\alpha+\frac{\sqrt{n}}{T} \right].
	\end{align} 
	With this we rewrite the summation in the bound as,
	\begin{equation}
	\langle f(t) \rangle_T \leq \gamma \sum_{n=0}^{\infty} r^{-n^2} \sum_{\alpha}p_\alpha \left( \sum_{G_\beta \in I_+} p_\beta + \sum_{G_\beta \in I_-} p_\beta \right).
	\end{equation}
	Given that $\sqrt{n} \leq n$ for all integer $n$, one can expand the length of the interval, so that
	\begin{equation}
	\langle f(t) \rangle_T \leq \gamma \sum_{n=0}^{\infty} r^{-n^2} \sum_{\alpha}p_\alpha \left( \sum_{G_\beta \in I'_+} p_\beta + \sum_{G_\beta \in I'_-} p_\beta \right) ,
	\end{equation}
	where we define
	\begin{align}
	I'_+ =  \left[ G_\alpha+ \frac{n}{T}, G_\alpha+\frac{n+1}{T} \right), \\
	I'_- = \left( G_\alpha -  \frac{n+1}{T}, G_\alpha+\frac{n}{T} \right].
	\end{align}
	It is straightforward to see that $\sum_{G_\beta \in I'_+} p_\beta + \sum_{G_\beta \in I'_-} p_\beta  \le 2 \xi_p\left(\frac{1}{T}\right)$, which allows us to further bound the summation by using Def. \ref{def:xi},
	\begin{equation}
	\langle f(t) \rangle_T \leq 2 \gamma \xi_p\left(\frac{1}{T}\right) \sum_{n=0}^\infty r^{-n^2}.
	\end{equation}
	To complete the proof we find the $\alpha$ that minimizes
	\begin{equation}
	\kappa(\alpha) = \gamma \sum_{n=0}^\infty r^{-n^2} = \gamma \left( \Theta_3(0,r^{-1})+1  \right),
	\end{equation}
	which occurs for $\alpha \approx 0.6347$. This gives $\kappa \approx 2.8637<3$. Finally,
	\begin{equation}
	\langle f(t) \rangle_T \leq 3 \pi \xi_p\left(\frac{1}{T}\right),
	\end{equation} 
	which completes the proof.

\end{proof}

\subsection{Proof of equilibration timescales of correlation functions}
\label{app:timescaleCorr}

\begin{customthm}{3}\label{thm:timescalesCorr}
	For any Hamiltonian $H$ and state $\rho$ such that $[H,\rho] = 0$, and any observable $A$, the time correlation function $C^A(t) = \tr{\rho A(t) A }$ satisfies
	\begin{align}
	\label{eq:boundCorrAveraged2}
	\frac{1}{T}\int_0^T \frac{ |C^A(t) - C^A_\infty|^2  }{\left(C^A(0)\right)^2} dt  &\le 3 \pi    \left( \frac{a(\epsilon)}{\sigma_G} \frac{1}{T} + \delta(\epsilon) \right)   
	\end{align}
	for any $\epsilon>0$. Here, $a(\epsilon)$ and $\delta(\epsilon)$ are as in Definition~\ref{def:xi} for the normalized distribution $v_\alpha \equiv \frac{\rho_{jj} |A_{jk}|^2}{C^A(0)}$, and 
	$\sigma_G$ is given by
	\begin{align}
	\sigma_G^2 &= \frac{ 1 }{ C^A(0) } \tr{\rho [A ,H] [H,A ]}  - \frac{\tr{\rho [H,A ] A }^2}{\left(C^A(0)\right)^2}.
	\end{align}
\end{customthm}

\begin{proof}
	Following Lemma \ref{lem:avgbound}, we have that
	\begin{align} \label{eq:testedbound}
	\frac{1}{T}\int_0^T \frac{ |C_t - C_\infty|^2  }{(C^A(0))^2} dt &\le 3 \pi   \xi_v\left( \tfrac{1}{T}\right),   
	\end{align}
	where $\xi_v\left( \tfrac{1}{T}\right)$ is taken with respect to the distribution $v_\alpha$, which is normalized:
	\begin{align}
	\sum v_\alpha~=~\sum_{jk} \frac{\rho_{jj} |A_{jk}|^2}{C^A(0)}~=~\frac{\tr{\rho  A^2}}{C^A(0)} = 1.
	\end{align}
	
	It was shown in~\cite{Garcia-PintosPRX2017} (Proposition 5) that the function $\xi_p(x)$ satisfies 
	\begin{align}
	\label{eqaux:xibound}
	\xi_p(x) \le \frac{a(\epsilon)}{\sigma_G} x + \delta(\epsilon),       
	\end{align}
	with $a(\epsilon) = \frac{ \xi_p(\epsilon) }{\epsilon} \sigma_G$ and $\delta(\epsilon) = \xi_p(\epsilon)$, and where $\varepsilon \in (0,\infty)$.
	Finally, the standard deviation of the distribution $p_v$ is 
	\begin{align}
	\sigma_G^2 &=  \sum_\alpha p_\alpha G_\alpha^2 - \left( \sum_\alpha p_\alpha G_\alpha \right)^2  \nonumber \\
	& = \sum_{jk} \frac{ \rho_{jj} |A_{jk}|^2 }{ C^A(0) } (E_j - E_k)^2   - \left( \sum_{jk} \frac{ \rho_{jj} |A_{jk}|^2 }{ C^A(0) } (E_j - E_k)   \right)^2 \nonumber \\
	& = \frac{ 1 }{ C^A(0) } \tr{\rho [A ,H] [H,A ]} - \frac{\tr{\rho [H,A ] A }^2}{(C^A(0))^2},
	\end{align}
	which completes the proof.
\end{proof}

\subsection{Proof of equilibration timescales of symmetric correlation functions}
\label{app:timescaleCorrSymm}
\begin{customthm}{4}
	For any Hamiltonian $H$ and state $\rho$ such that $[H,\rho] = 0$, and any observable $A$, the time correlation function $C_s^A(t) = \tr{\rho \{A, A(t)\} }$ satisfies
	\begin{align}
	\frac{1}{T}\int_0^T \frac{ |C^A_s(t) - C^A_{s,\infty}|^2  }{(C^A_s(0))^2} dt  &\le 3 \pi    \left( \frac{a(\epsilon)}{\sigma_G} \frac{1}{T} + \delta(\epsilon) \right),   
	\end{align}
	for any $\epsilon>0$. Here, $a(\epsilon)$ and $\delta(\epsilon)$ are as in Definition \ref{def:xi}
	for the normalized distribution $v_\alpha^s \equiv \frac{\rho_{jj} + \rho_{kk}}{2} \frac{|A_{jk}|^2}{C^A_s(0)}$, and  
	\begin{align}
	\sigma_G^2 &= \frac{ 1 }{ C^A_s(0) } \tr{\rho [A_0,H] [H,A_0]}.
	\end{align}
	
\end{customthm}

\begin{proof}
	The symmetric correlation function is defined as
	\begin{align}
	C^A_s(t) \equiv \frac{1}{2} \tr{ \rho \left\{ A , A(t) \right\}}  = \frac{C^A(t) + C^A(t)^*}{2}.
	\end{align}
	The equivalent of Eq.\ref{eq:correlationfunction} 
	becomes
	\begin{align}
	\frac{C^A_s(t)} {C^A_s(0) } &= \sum_{j k} \frac{\rho_{jj} + \rho_{kk}}{2} \frac{|A_{jk}|^2}{C^A_s(0) } e^{-i(E_j - E_k)t},
	\end{align}
	The proof of Theorem \ref{thm:timescalesCorrSym}
	is identical as the previous proof, with the symmetrized distribution $v_\alpha^S \equiv \frac{\rho_{jj} + \rho_{kk}}{2} |A_{jk}|^2$.
	In this case the variance of the normalized distribution $v_\alpha^S$ becomes
	\begin{align}
	\sigma_G^2 &=  \sum_\alpha p_\alpha^S G_\alpha^2 - \left( \sum_\alpha p_\alpha^S G_\alpha \right)^2   \nonumber \\
	& = \sum_{jk} \frac{ \frac{\rho_{jj} + \rho_{kk}}{2} |A_{jk}|^2 }{ C^A_s(0) } (E_j - E_k)^2   - \left( \sum_{jk} \frac{ \frac{\rho_{jj} + \rho_{kk}}{2} |A_{jk}|^2 }{ C^A_s(0) } (E_j - E_k)   \right)^2   \nonumber \\
	& = \frac{ 1 }{ C^A_s(0) } \tr{\rho [A ,H] [H,A ]}.
	\end{align}
\end{proof}

Equilibration then occurs within a timescale
\begin{align}
T_{eq}  &=   \frac{\kappa \pi \, a(\epsilon) \, \sqrt{C^A_s(0)}}{\sqrt{ \tr{\rho [A ,H] [H,A ]} }}, 
\end{align}
The denominator in $T_{eq}$ can be identified as an ``acceleration'' of the symmetric autocorrelation function. 
Indeed, 
\begin{align}
\frac{d^2 C^A_s(t) }{dt^2} = - \sum_{j k} (E_j - E_k)^2 \frac{\rho_{jj}+\rho_{kk}}{2} |A_{jk}|^2 e^{-i(E_j - E_k)t}.
\end{align}
Then, the equilibration timescale is
\begin{align}
T_{eq} =     \frac{\kappa \pi \, a(\epsilon) \, \sqrt{C^A_s(0)}}{\sqrt{ \left| \frac{d^2 C^A_s(t) }{dt^2} \big\vert_0    \right| }}. 
\end{align}

\subsubsection{Short-time evolution of symmetric correlation functions}

The symmetric autocorrelation functions is given by
\begin{align}
C^A_s(t) &= \sum_{j k} \frac{\rho_{jj} + \rho_{kk}}{2}  |A_{jk}|^2 e^{-i(E_j - E_k)t}.
\end{align}
Taking the Taylor expansion of $C^A_s(t)$,
\begin{align}
C^A_s(t) &= C^A_s(0)  - \left( i \sum_{j k} (E_j - E_k  ) 
\frac{\rho_{jj} + \rho_{kk}}{2}      |A_{jk}|^2 \right) t \nonumber \\
&- \left(  \sum_{j k} (E_j - E_k)^2  \frac{\rho_{jj} + \rho_{kk}}{2}  |A_{jk}|^2  \right) \frac{t^2}{2} +\mathcal{O}(t^3) \nonumber \\
&=C^A_s(0)- \frac{d^2 C^A_s }{dt^2} \bigg\vert_0 \frac{ t^2 }{2}  +\mathcal{O}(t^3) \nonumber \\
&=C^A_s(0) \left(  1 -    \frac{1}{2 C_0 } \frac{d^2 C^A_s }{dt^2} \bigg\vert_0   t^2  \right)  +\mathcal{O}(t^3).
\end{align}
For early times, the above expression decays on a timescale 
\begin{align}
\tau \equiv \frac{\sqrt{2}\sqrt{C^A_s(0)}}{ \sqrt{ \left| \frac{d^2 C^A_s(t) }{dt^2} \Big\vert_0    \right| } }.
\end{align}

\subsection{Proof of equilibration timescales of Kubo correlation functions}
\label{app:timescaleKubo}

\begin{customthm}{5}
	For any Hamiltonian $H$, thermal state $\rho$, and any observable $A$, the Kubo correlation function $C_{\textnormal{Kubo}}$ satisfies
	\begin{align}
	\label{eq-app:boundKubo}
	\frac{1}{T}\int_0^T \frac{ |C_{\textnormal{Kubo}}(t) - C_{\textnormal{Kubo},\infty} |^2  }{C_{\textnormal{Kubo}}(0)^2} dt  &\le 3 \pi    \left( \frac{a(\epsilon)}{\sigma_G} \frac{1}{T} + \delta(\epsilon) \right),   
	\end{align}
	for any $\epsilon>0$. Here, $a(\epsilon)$ and $\delta(\epsilon)$ are as in Definition~\ref{def:xi} 
	for the normalized distribution $w_\alpha \equiv \frac{e^{-\beta E_k}-e^{-\beta E_j}}{E_j-E_k} \frac{ \vert A_{jk}\vert^2 }{C_{\textnormal{Kubo}}(0)} $, and 
	\begin{align}
	\sigma_G^2 &= \frac{ 1 }{ \Ckubo(0) } \tr{ [A,\rho] [A,H]}.
	\end{align}
\end{customthm}

\begin{proof}
	The Kubo correlation function can be written as
	\begin{align}
	C_{\text{Kubo}}(t) \propto \sum_{j\neq k}\frac{e^{-\beta E_k}-e^{-\beta E_j}}{E_j-E_k} \vert A_{jk}\vert^2 e^{i t(E_j-E_k)},
	\end{align}
	with the proportionality constant defined by $ C_{\text{Kubo}}(0) = 1$.
	We can then write
	\begin{align}
	\frac{C_{\text{Kubo}}(t)-C_{\text{Kubo}}(\infty)}{C_{\text{Kubo}}(0)} \nonumber &=\frac{\sum_{j \neq k}\frac{e^{-\beta E_k}-e^{-\beta E_j}}{E_j-E_k} \vert A_{jk}\vert^2 e^{i t(E_j-E_k)}}{\sum_{j,k}\frac{e^{-\beta E_k}-e^{-\beta E_j}}{E_j-E_k} \vert A_{jk}\vert^2} 
	\\ & \equiv \frac{\sum_{\alpha \neq 0} w_\alpha e^{-i t G_\alpha}}{\sum_\alpha w_\alpha},
	\end{align}
	where we define $w_\alpha \equiv \frac{e^{-\beta E_k}-e^{-\beta E_j}}{E_j-E_k} \vert A_{jk}\vert^2\ge 0$.
	Given that $w_\alpha \ge 0$, we can perform similar calculations as for Theorem~\ref{thm:timescalesCorr}, 
	albeit with a different probability distribution. 
	Thus, we also have
	\begin{align}
	\left \langle\left(\frac{C_{\text{Kubo}}(t)-C_{\text{Kubo}}(\infty)}{C_{\text{Kubo}}(0)} \right)^2 \right \rangle_T \le 3 \pi \xi_w (1/T).
	\end{align}
	Defining the normalized distribution $q_\alpha \equiv w_\alpha/\Ckubo(0)$. The variance is 
	\begin{align}
	\sigma_G^2 &=  \sum_\alpha q_\alpha G_\alpha^2 - \left( \sum_\alpha q_\alpha G_\alpha \right)^2  \nonumber \\
	& = \sum_{jk} \frac{1}{\Ckubo(0)} \frac{ e^{-\beta E_k}-e^{-\beta E_j} } {E_j-E_k} \vert A_{jk}\vert^2  (E_j - E_k)^2   - \left( \sum_{jk}  \frac{1}{\Ckubo(0)} \frac{ e^{-\beta E_k}-e^{-\beta E_j} } {E_j-E_k} \vert A_{jk}\vert^2 (E_j - E_k)   \right)^2 \nonumber \\
	& = \frac{1}{\Ckubo(0)} \sum_{jk}   \left(e^{-\beta E_k}-e^{-\beta E_j} \right)   \vert A_{jk}\vert^2 (E_j - E_k)    \nonumber \\
	& = \frac{ 1 }{ \Ckubo(0) } \tr{ [A,\rho] [A,H]} ,
	\end{align}
	completing the proof.
\end{proof}

\subsection{Scaling of $a$ and $\delta$}
\label{app:aanddelta}

The proofs of Theorems~(3-5)
rely on the fact that the function 
$\xi_p(x)$,
defined for any normalized distribution $p_\alpha$ as  the maximum distribution that fits an interval $x$ 
\begin{align}
\xi_p(x) \equiv \max_{\beta} \sum_{\alpha : G_\beta \in \left[ G_\beta, G_\beta + x \right] } p_\alpha,
\end{align}
satisfies
\begin{align}
\label{eqaux:xibound}
\xi_p(x) \le \frac{a(\epsilon)}{\sigma_G} x + \delta(\epsilon),       
\end{align}
which was shown in~\cite{Garcia-PintosPRX2017} (Proposition 5). Here $a(\epsilon) \equiv \frac{ \xi_p(\epsilon) }{\epsilon} \sigma_G$ and $\delta(\epsilon) \equiv \xi_p(\epsilon)$, where
$\sigma_G$ is the standard deviation of the distribution $p_\alpha$.
The function $a(\epsilon)$ ends up in the bound of the equilibration timescales, as $T_{eq} =     \frac{\pi \, a(\epsilon)}{\sigma_G}$, while $\delta(\epsilon)$ governs the long time behavior in the bounds. 

Given that $\xi_p(x)$ characterizes how much of the distribution $p_\alpha$ fits an interval $x$, the value of $a$ in Eq.~\eqref{eqaux:xibound} depends on how well $1/\sigma_G$ serves to characterize the region where the distribution $p_\alpha$ is supported.
Roughly speaking, whenever $1/\sigma_G$ is a good estimate of the width of such small region, then one expects $a \sim \mathcal{O}(1)$. This is well illustrated when considering a unimodal distribution (e.g. a Gaussian). In such case, the fraction of the distribution that fits an interval $x$ is roughly $x$ times the width $1/\sigma_G$ of the window where the distribution is supported, and $\xi_p(x) \sim x/\sigma_G$, so that $a \sim \mathcal{O}(1)$.
Multimodal distributions violate such condition, as for them the standard deviation does not characterize the regions in which the distribution has considerable support.
At the same time, $\delta$ in Eq.~\eqref{eqaux:xibound} carries information of the fine structure of $p_\alpha$, indicating the scale at which the distribution can no longer be coarse-grained to a continuous distribution. The only way that $\delta \ll 1$ fails is for distributions that are not smooth, in which a small region of width $\epsilon$ is significantly populated. Thus, for distributions that are smooth in a coarse-grained sense, and approximately unimodal, one expects to be able to find a small enough $\epsilon$  such that $\delta(\epsilon )  \ll 1 $ and $a(\epsilon) \sim \mathcal{O}(1)$.

In summary, the problem of proving fast equilibration timescales in our approach can thus be linked to knowing whether the relevant distribution $p_\alpha$ is `approximately unimodal'.
We argue that for an strongly interacting many-body system this will typically be the case. 
Consider for instance the case of Theorem~\ref{thm:timescalesCorr},
where the relevant distribution is given by $v_\alpha \equiv \frac{\rho_{jj} |A_{jk}|^2}{C^A(0)}$.

The large number of energy gaps present in a many-body system implies a dominance of small gaps over larger ones, which favors that, on a coarse-grained sense, the distribution over gaps shows a decay as the size $|G_\alpha|$ of the gap increases. This is reinforced by the tendency of off-diagonal matrix elements $|A_{jk}|$ of local observables to decay as the levels considered are further apart.
Existing numerical results on off-diagonal matrix elements of local observables in non-integrable models are consistent with all the requirements listed here \cite{Beugeling15,Mondaini17}.
The present arguments suggest distributions $v_\alpha$ that decay for larger values of $|G_\alpha|$, and are therefore unimodal, and also smoothly distributed. This is confirmed in the simulations in Appendix~\ref{app:simulation} in a non-integrable model on Fig.~\ref{fig:hist} (left), and to a somewhat lesser extent in an integrable model too on Fig.~\ref{fig:hist} (right).

\section{Simulations}
\label{app:simulation}
To calculate the function given in Definition~\ref{def:xi} in the main text
exactly one needs to find the maximum sum of $p_\alpha$ such that $\alpha:G_\alpha \in [G_\lambda, G_\lambda+x]$. This calculation scales quite unfavourably with system size. If we have $N$ energies, the number of intervals one must probe is quadratic in $N$. For each $x \sim \mathcal{O}(10^{-1})$ the intervals near the center are quite dense, making the entire algorithm for one choice of $x$ approximately  scale like  $\mathcal{O}(N^3)$. For this reason, we exactly diagonalize the Hamiltonian given in Eq.~\ref{eq:hamiltonian}, 
and numerically approximate $\xi(x)$. This is done by means of a Monte Carlo scheme where we randomly select intervals defined by $G_\lambda$ using a normal distribution defined by $\mu_G = \sum_\alpha p_\alpha G_\alpha$ and $\sigma_G$ given in Definition~\ref{def:xi}. 

\begin{figure*}[h!]
	\centering
	\includegraphics[width=0.49\linewidth]{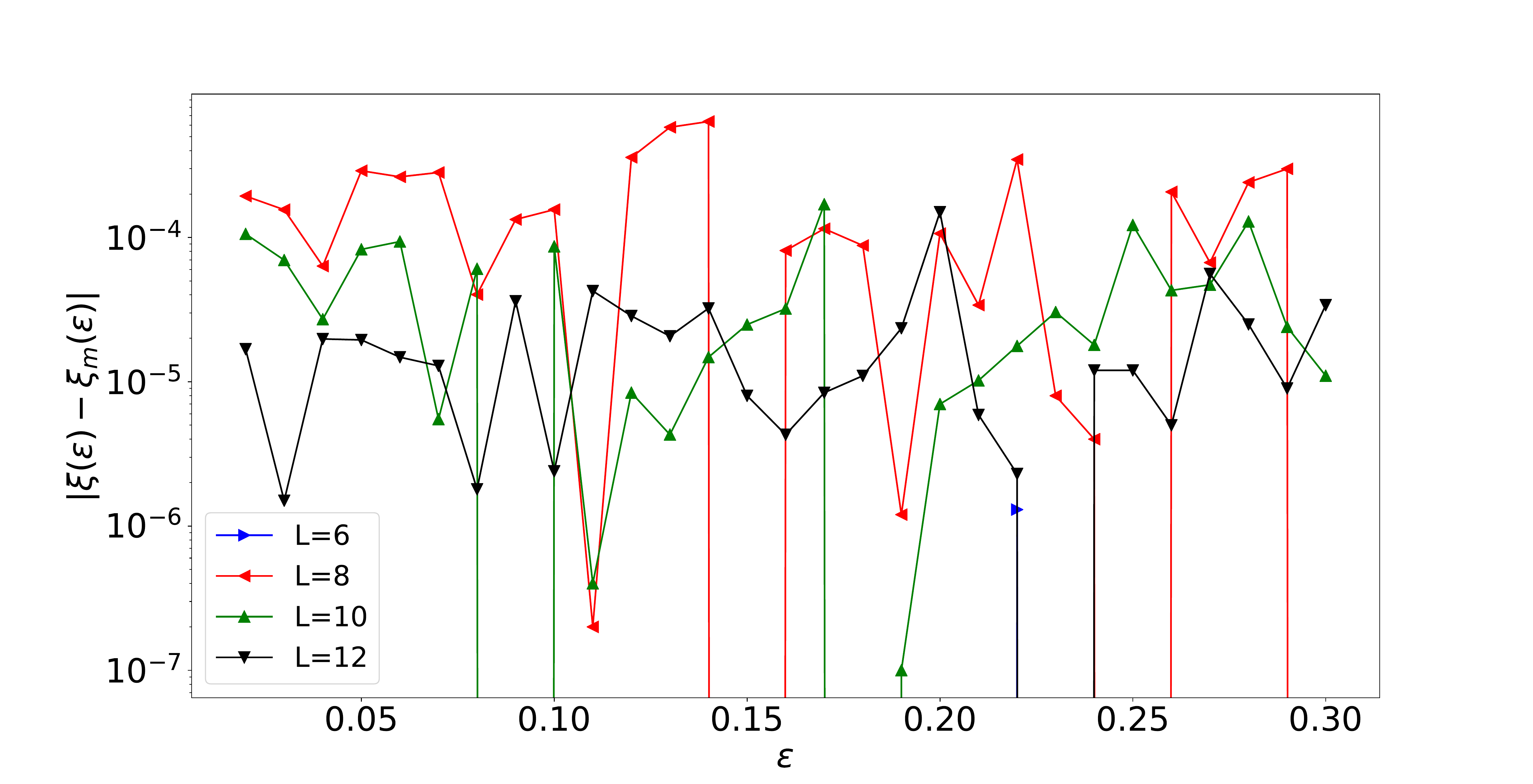}
	\includegraphics[width=0.49\linewidth]{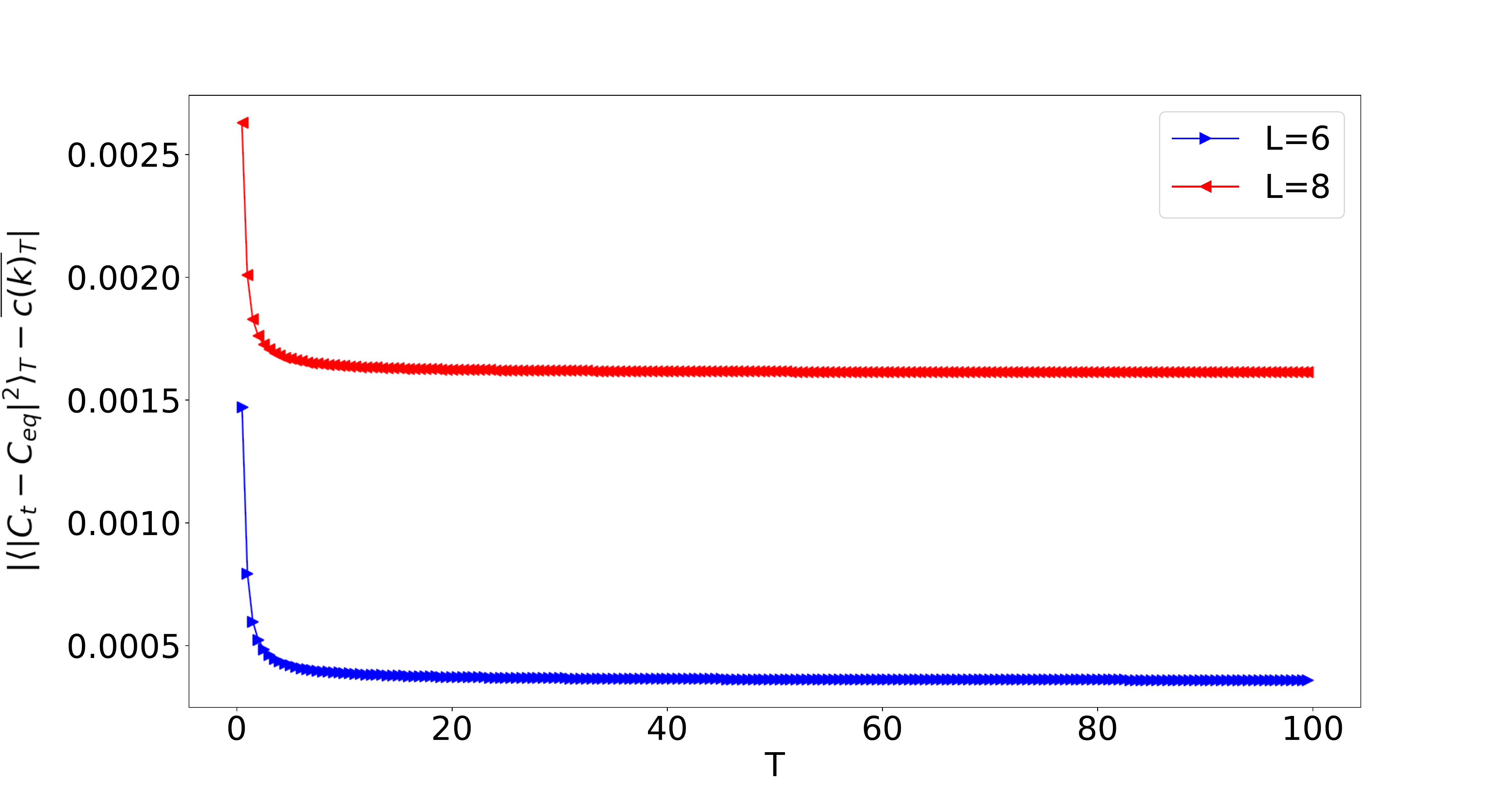}
	\caption{(left) Forward error plot depicting the accuracy of the Monte Carlo scheme at different values of system size. $\xi(\epsilon)$ is exactly calculated and $\xi_m(\epsilon)$ is calculated using the Monte Carlo scheme with 10,000 samples per value of $\epsilon$.
		(right) Forward error plot depicting the accuracy of the integral approximation in calculating the left hand side of equation \ref{eq:boundCorrAveraged2} at various values of time and system size. A step size in time of $\Delta t = 0.001$ was used.
	}
	\label{fig:numerical-checks}
\end{figure*}
Figure \ref{fig:numerical-checks} (left) depicts the accuracy of this scheme. Unsurprisingly $L=6$ is exactly calculated and is not visible on the except at one location. The other cases show the approximation scheme performs better at larger system sizes. Despite this improvement, the accuracy of the scheme roughly puts us accurate to the fourth digit in all cases, making this scheme more than accurate enough. 
Quantities such as $\mu_G$ and $\sigma_G$ can be calculated exactly given the exact diagonalization. However the left hand side of Eq.~\ref{eq:boundCorrAveraged}
in the main text has a time order complexity of $\mathcal{O}(N^4)$, making it again extremely difficult to calculate exactly. To get around this, we simply define a grid $t_k = \Delta k$ where $k=0,1,2 \dots$ and average over the values calculated of $|C^A(t_k)-C_{\infty}^A|^2$ as, 
\begin{equation}
\overline{c(k)}_T = \frac{1}{k+1} \sum_{i=0}^k |C^A(t_i)-C_{\infty}^A|^2.
\end{equation}

Figure \ref{fig:numerical-checks} (right)  shows the forward error of this scheme, showing an expected first order accuracy in time. Since we have an accuracy which is satisfactory for the scales we are comparing with the bound which tends have a roughly $10^{-1}$ disagreement between the two sides of the bound. Finally to the optimal choice of $a(\epsilon)$ and $\delta(\epsilon)$. 
For the plots present in Figure~\ref{fig:boundplot} in the main text
we simply took the smallest $\delta(\epsilon)$ available to minimize the resolution of our bound and picked the corresponding $a(\epsilon)$. This choice has an obvious issue in the $L=8$ case but begins to be more favourable in the larger system sizes. Looking at Figure~\ref{fig:deltaplotexact}
we see the value $a(\epsilon)$ can grow quite quickly due to finite size effects, making the prefactor outside the $1/T$ term quite large.

\subsection{Integrable models}\label{app:integrable}
Next, this section provide an example of showing how our bounds on timescales are affected in integrable models, highlighting the negative effect of degeneracies. Suppose we choose to define our Hamiltonian of Eq.~\eqref{eq:hamiltonian}
in the main text with parameters $H = (-0.5,0, -0.5, 0 )$. This corresponds to an Ising model with a transverse field. The issue in general with this model comes from investigating the behaviour of the corresponding $\delta(\epsilon)$ and the fact that the frequencies $G_\alpha$ are very degenerate, meaning this function will not necessarily decay to zero as we take $\epsilon \to 0$.

\begin{figure}[h]
	\centering
	\includegraphics[width = 0.49 \linewidth]{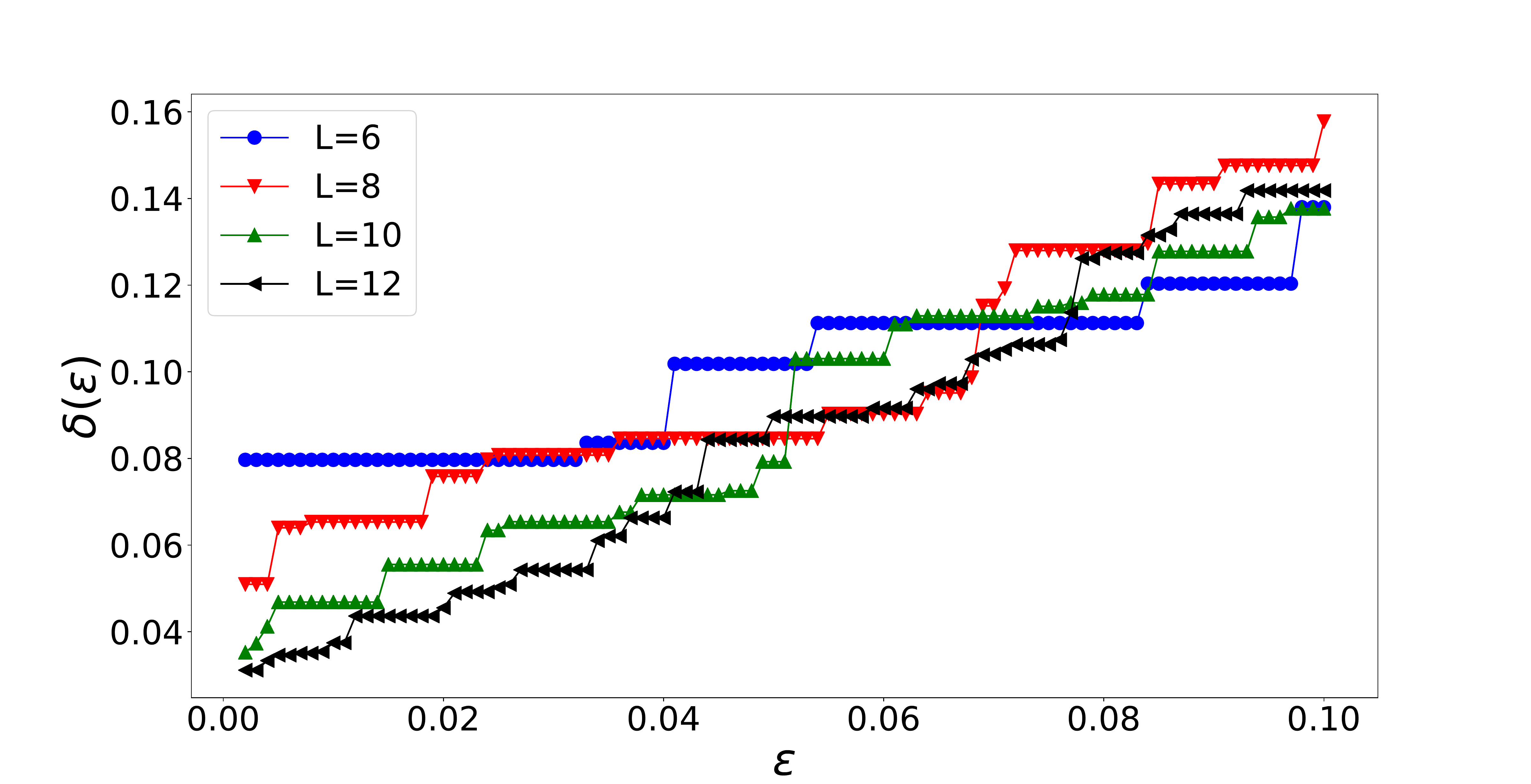}
	\includegraphics[width = 0.49\linewidth]{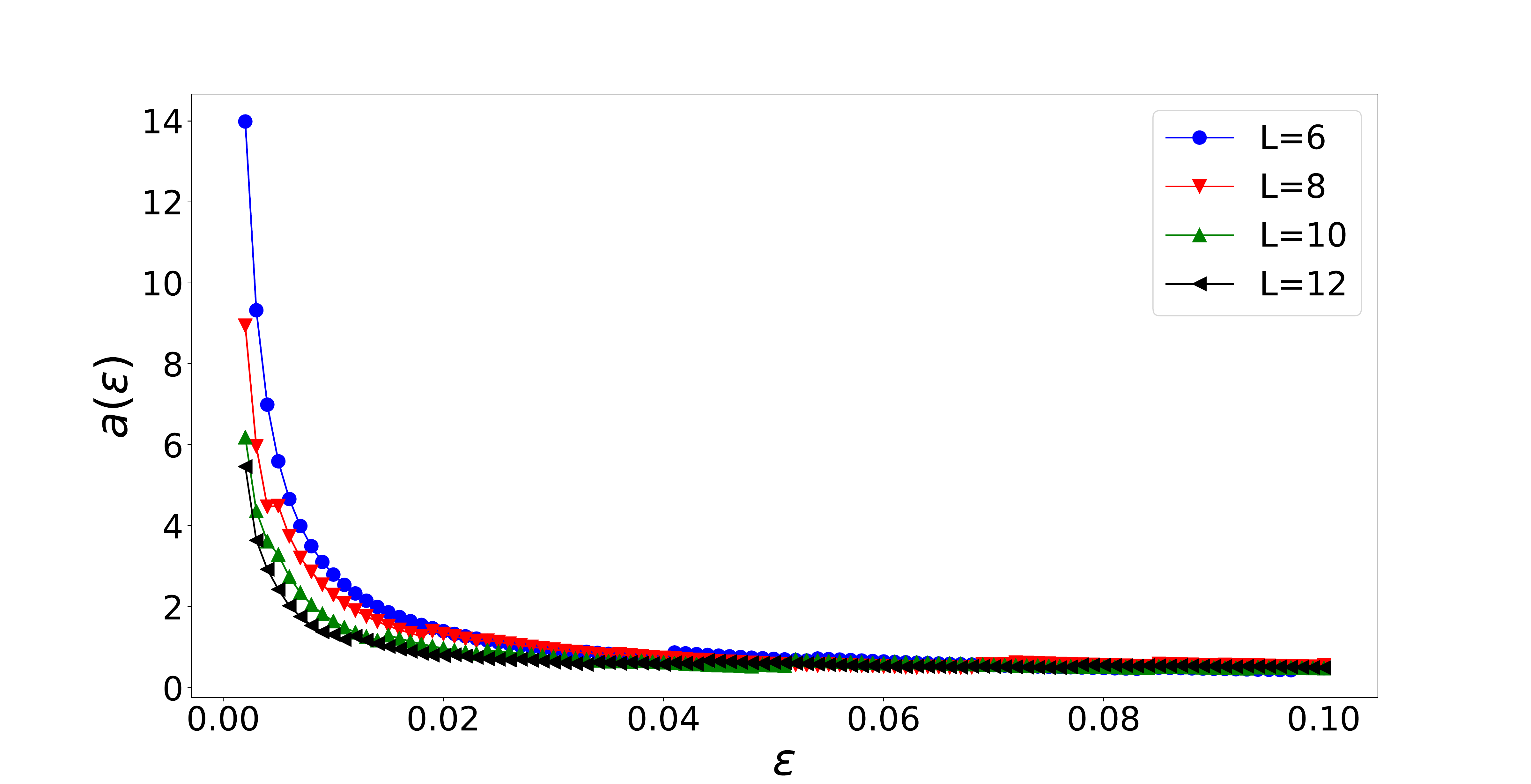}
	\caption{Plots of $\delta(\epsilon)$ (top) and $a(\epsilon)$ (bottom) for distribution $v_\alpha$ in Theorem~\ref{thm:timescalesCorr}, obtained by exact diagonalization, and a Monte Carlo approximation of the function from Definition~\ref{def:xi}. The plots share the same x-axis and were generated with 10,000 sampled frequency intervals. 
		Small values of $\delta$ imply equilibration occurs for long enough times, while the value of $a$ controls the prefactor in the equilibration timescale Eq.~\eqref{eq:boundCorrAveraged2}. } 
	\label{fig:deltaplotint}
	\label{fig:aplotint}
\end{figure}
In Figure \ref{fig:deltaplotint} we see the issue emerging with the bound found in Theorem~\ref{thm:timescalesCorr}. 
The degeneracy of the $G_\alpha$ terms cause the decrease in $\delta(\epsilon)$ to happen in discrete steps triggered by calculating $\delta(\epsilon)$ in a small enough region to differentiate two degenerate values of $G_\alpha$ which are close. Thus at small $\epsilon$ we still expect our resolution of equilibrium to be quite large. This slow decay of $\delta(\epsilon)$ also causes $a(\epsilon)$ to become quite large very quickly, as one needs $\delta(\epsilon)$ to be roughly linear for $a(\epsilon)$ to be reasonably small. This suggests that perhaps alternative approaches are required to bound the equilibration of two point time correlation functions in integrable models.

\subsection{Distribution of $v_\alpha \equiv \frac{\rho_{jj} |A_{jk}|^2}{C_0}$}

Finally, we show the distributions of $v_\alpha \equiv \frac{\rho_{jj} |A_{jk}|^2}{C_0}$ and comment on the differences between the integrable case  and the case that obeys the ETH. To proceed we define a coarse grained version $\bar{p}_\alpha$ of $v_\alpha \equiv \frac{\rho_{jj} |A_{jk}|^2}{C_0}$, where we define $n$ bins and bins, 
$b_1 = [G_{min}, G_{min}+ \Delta G], b_2 = [ G_{min} + \Delta G, G_{min} + 2\Delta G ], \dots $, where $\Delta G = \frac{G_{max}-G_{min}}{n}$. The coarse grained probability is then obtained by summing the associated probabilities, $\bar{p}_\alpha = \sum_{G_\beta \in b_\beta} v_\beta$.

\begin{figure}[h!]
	\centering
	\includegraphics[width=0.49\linewidth]{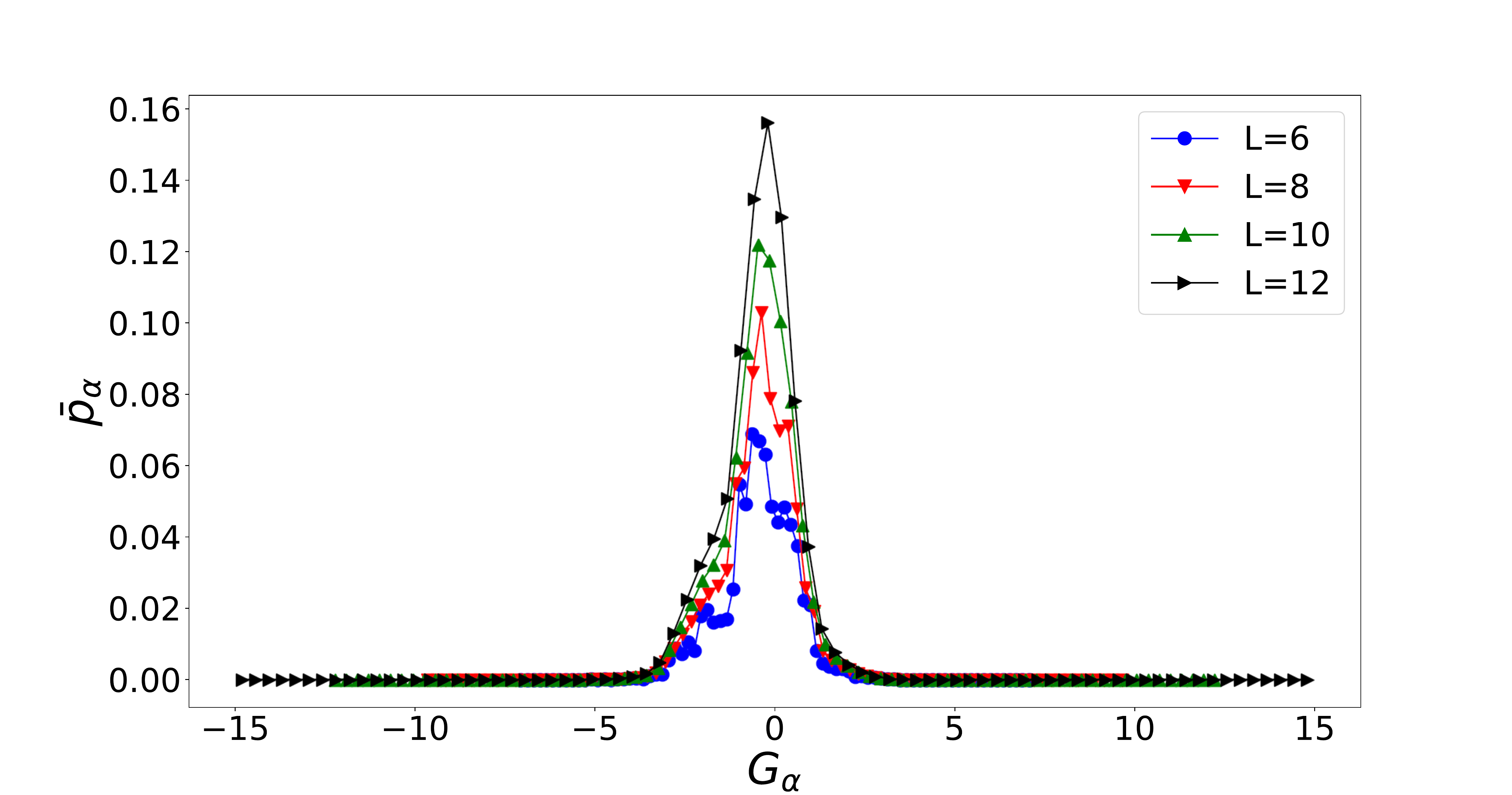}
	\includegraphics[width=0.49\linewidth]{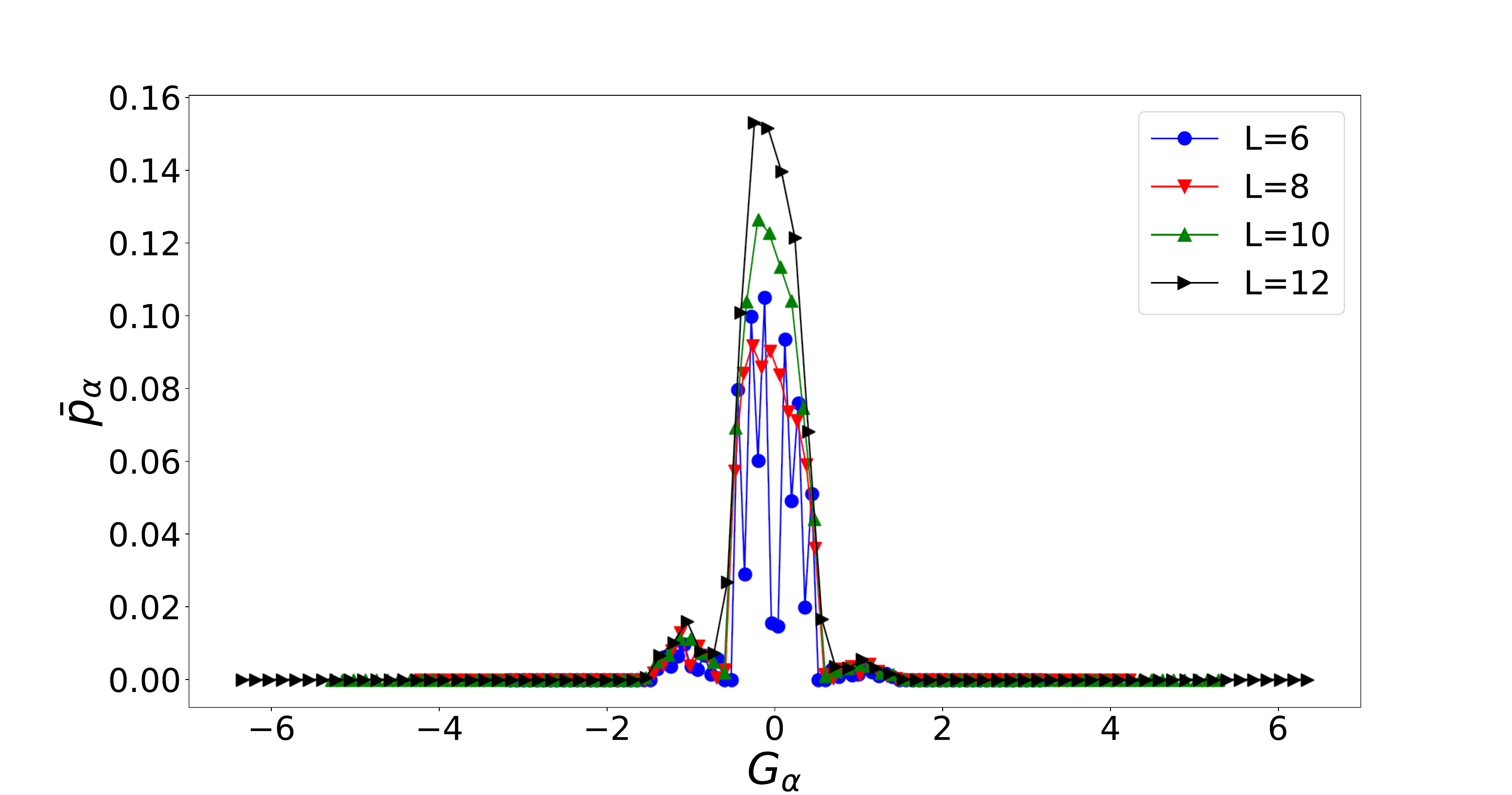}
	\caption{Plot of $\bar{p}_\alpha$ obtained from coarse-graining  $v_\alpha \equiv \frac{\rho_{jj} |A_{jk}|^2}{C_0}$ against frequency, with $n=80$ bins at various system sizes. The case for which ETH is satisfied is featured on the left, while the  integrable case is on the right. For the latter the distribution is less uni-modal, which leads to larger values of $a(\epsilon)$, as depicted on Fig.~\ref{fig:aplotint} (right).}
	\label{fig:hist}
\end{figure}

The result is given in Figure~\ref{fig:hist}. The ETH case approaches a unimodal distribution quicker than the integrable case, however both distributions appear favorable in the coarse grained probabilities. Note, increasing the number of bins significantly did not significantly affect the shape of the curve.

\end{document}